\documentclass[journal,comsoc]{IEEEtran}

\usepackage{amsmath,amssymb,amsthm}
\usepackage{xifthen}
\usepackage{mathtools}
\usepackage{enumerate}
\usepackage{microtype}
\usepackage{xspace}
\usepackage{bm}
\usepackage[T1]{fontenc}
\usepackage{fancyhdr}
\usepackage{lastpage}
\usepackage{bbm}

\newcommand{\mbs}[1]{\pmb{#1}}
\newcommand{\vect}[1]{{\lowercase{\mbs{#1}}}}
\newcommand{\mat}[1]{{\uppercase{\mbs{#1}}}}

\newcommand{\T}{{\scriptscriptstyle\mathsf{T}}}
\renewcommand{\H}{{\scriptscriptstyle\mathsf{H}}}

\renewcommand{\Re}[1][]{\ifthenelse{\isempty{#1}}{\operatorname{Re}}{\operatorname{Re}\left(#1\right)}}
\renewcommand{\Im}[1][]{\ifthenelse{\isempty{#1}}{\operatorname{Im}}{\operatorname{Im}\left(#1\right)}}

\newcommand{\SNR}{\mathsf{snr}}

\newcommand{\dv}{\vect{d}}
\newcommand{\ev}{\vect{e}}

\newcommand{\iv}{\vect{i}}

\newcommand{\rv}{\vect{r}}
\newcommand{\sv}{\vect{s}}

\newcommand{\uv}{\vect{u}}
\newcommand{\vv}{\vect{v}}

\newcommand{\xv}{\vect{x}}
\newcommand{\yv}{\vect{y}}
\newcommand{\zv}{\vect{z}}

\newcommand{\muv}{\vect{\mu}}
\newcommand{\alphav}{\vect{\alpha}}
\newcommand{\betav}{\vect{\beta}}
\newcommand{\gammav}{\vect{\gamma}}

\newcommand{\phiv}{\vect{\phi}}
\newcommand{\Sigmam}{\pmb{\Sigma}}
\newcommand{\Lambdam}{\pmb{\Lambda}}

\newcommand{\Xim}{\pmb{\Xi}}

\newcommand{\Cm}{\mat{c}}

\newcommand{\Mm}{\mat{M}}

\newcommand{\Qm}{\mat{q}}
\newcommand{\Rm}{\mat{r}}
\newcommand{\Sm}{\mat{s}}

\newcommand{\Um}{\mat{u}}
\newcommand{\Vm}{\mat{V}}

\newcommand{\Ym}{\mat{y}}

\newcommand{\Dc}{{\mathcal D}}

\newcommand{\Nc}{{\mathcal N}}

\newcommand{\Pc}{{\mathcal P}}

\newcommand{\Sc}{{\mathcal S}}
\newcommand{\Tc}{{\mathcal T}}

\newcommand{\Vc}{{\mathcal V}}
\newcommand{\Xc}{{\mathcal X}}

\newcommand{\CC}{\mathbb{C}}

\newcommand{\Id}{\mat{\mathrm{I}}}

\newcommand{\CN}[1][]{\ifthenelse{\isempty{#1}}{\mathcal{N}_{\mathbb{C}}}{\mathcal{N}_{\mathbb{C}}\left(#1\right)}}

\renewcommand{\P}[1][]{\ifthenelse{\isempty{#1}}{\mathbb{P}}{\mathbb{P}\left(#1\right)}}
\newcommand{\E}[1][]{\ifthenelse{\isempty{#1}}{\mathbb{E}}{\mathbb{E}\left[#1\right]}}
\newcommand{\I}[1][]{\ifthenelse{\isempty{#1}}{\mathbb{I}}{\mathbb{I}\left\{#1\right\}}}
\renewcommand{\det}[1][]{\ifthenelse{\isempty{#1}}{\mathrm{det}}{{\rm det}\left(#1\right)}}
\newcommand{\trace}[1][]{\ifthenelse{\isempty{#1}}{{\rm tr}}{{\rm tr}\left\{#1\right\}}}
\newcommand{\rank}[1][]{\ifthenelse{\isempty{#1}}{\mathrm{rank}}{{\rm rank}\left(#1\right)}}
\newcommand{\diag}[1][]{\ifthenelse{\isempty{#1}}{\mathrm{diag}}{{\rm diag}\left(#1\right)}}
\newcommand{\Cov}[1][]{\ifthenelse{\isempty{#1}}{\mathsf{Cov}}{\mathsf{Cov}\left(#1\right)}}

\newcommand{\defeq}{\triangleq}
\newcommand{\eqdef}{\triangleq}

\newtheorem{proposition}{Proposition}
\newtheorem{remark}{Remark}
\newtheorem{lemma}{Lemma}

\newcounter{enumi_saved}
\usepackage{answers}
\Newassociation{solution}{Solution}{solutionfile}

\AtBeginDocument{\Opensolutionfile{solutionfile}[\jobname]}
\AtEndDocument{\Closesolutionfile{solutionfile}\clearpage
}

\IfFileExists{MinionPro.sty}{
}{
}

\usepackage{subfigure}
\usepackage{graphicx}
\DeclareGraphicsExtensions{.eps, .jpg, .png , .pdf, .bmp} 
\usepackage{multirow}
\usepackage{array}
\usepackage{amsmath,bbm}
\usepackage[dvipsnames]{xcolor}
\usepackage{enumitem}
\usepackage{stfloats}
\usepackage[flushleft]{threeparttable}
\usepackage{pgfplots}
\pgfplotsset{minor grid style={dotted}}
\pgfplotsset{major grid style={dashed}}
\usepackage{epstopdf}
\usepackage{tikz}
\usetikzlibrary{fit,positioning,arrows,shapes,shapes.multipart,calc,arrows.meta}
\usepackage[ruled,vlined,linesnumbered]{algorithm2e}
\mathtoolsset{showonlyrefs}
\usepackage{grffile}
\pgfplotsset{compat=newest}
\usetikzlibrary{plotmarks}
\usetikzlibrary{arrows.meta}
\usepgfplotslibrary{patchplots}

\renewcommand{\rv}[1]{{\mathsf{#1}}}
\newcommand{\rvVec}[1]{{\pmb{\mathsf{#1}}}}
\newcommand{\rvMat}[1]{{\pmb{\mathsf{#1}}}}
\newcommand{\const}{c_0}

\newcommand{\of}[1]{^{(#1)}}
\newcommand{\ofH}[1]{^{(#1)\H}}
\newcommand{\ofT}[1]{^{(#1)\T}}
 \newcommand*\dif{\mathop{} \mathrm{d}}
 \newcommand{\vectr}[1]{{\rm vec} (#1)}

\renewcommand{\defeq}{:=}
\renewcommand{\eqdef}{=:}
\renewcommand{\SNR}{{\rm SNR}}
\renewcommand{\Id}{\mat{I}}

\newcommand{\LLR}{{\rm LLR}}

\newcommand{\m}[2]{m_{#1,#2}}
\newcommand{\gam}[2]{\gammav_{#1,#2}}

\newcommand{\new}{^{{\rm new}}}
\newcommand{\lnew}{_{{\rm new}}}
\newcommand{\Nalpha}{\mathfrak{N}_\alpha}

\newcommand{\Nbeta}{\mathfrak{N}_\beta}

\newcommand{\GMI}{R_{\rm GMI}}

\newcommand{\uvec}[1]{\ensuremath{\underline{\boldsymbol{#1}}}}

\newcommand{\revise}[1]{{#1}}

\begin{document}
	
\title{Multi-User Detection Based on Expectation Propagation for the Non-Coherent SIMO Multiple Access Channel}

\author{	
	Khac-Hoang Ngo,~\IEEEmembership{Student Member,~IEEE,} Maxime Guillaud,~\IEEEmembership{Senior Member,~IEEE,} \\ Alexis Decurninge,~\IEEEmembership{Member,~IEEE,} Sheng Yang,~\IEEEmembership{Member,~IEEE,} Philip Schniter,~\IEEEmembership{Fellow,~IEEE}%
	\thanks{Khac-Hoang Ngo is with Mathematical and Algorithmic Sciences Laboratory, Paris Research Center, Huawei Technologies, 92100 Boulogne-Billancourt, France, and also with Laboratory of Signals and Systems, CentraleSup\'elec, University of Paris-Saclay, 91190 Gif-sur-Yvette, France~(e-mail: \texttt{ngo.khac.hoang@huawei.com}).
	}%
	\thanks{Maxime Guillaud and Alexis Decurninge are with Mathematical and Algorithmic Sciences Laboratory, Paris Research Center, Huawei Technologies, 92100 Boulogne-Billancourt, France~(e-mail: \texttt{\{maxime.guillaud,alexis.decurninge\}@huawei.com}).
	}%
	\thanks{Sheng Yang is with Laboratory of Signals and Systems, CentraleSup\'elec, University of Paris-Saclay, 91190 Gif-sur-Yvette, France~(e-mail: \texttt{sheng.yang@centralesupelec.fr}).
	}%
	\thanks{Philip Schniter is with Department of Electrical and Computer Engineering, The Ohio State University, Columbus, OH 43210, USA~(e-mail: \texttt{schniter.1@osu.edu}).
	}%
	\thanks{This paper was presented in part at the 53rd Asilomar Conference on Signals,
		Systems, and Computers, Pacific Grove, CA, USA, 2019~\cite{HoangAsilomar2019_EP}.}%
}


\maketitle

\begin{abstract} 
	We consider the non-coherent single-input multiple-output~(SIMO) multiple access channel with \revise{general} signaling under \revise{spatially correlated} Rayleigh block fading. We propose a novel soft-output multi-user detector that computes an approximate marginal posterior of each transmitted signal using only the knowledge about the channel distribution. Our detector is based on expectation propagation~(EP) approximate inference and has polynomial complexity in the number of users, \revise{number of receive antennas and channel coherence time}. We also propose two simplifications of this detector with reduced complexity. 
	\revise{With Grassmannian signaling, the proposed detectors outperform a state-of-the-art non-coherent detector with projection-based interference \revise{mitigation}. With pilot-assisted signaling, the EP detector outperforms, in terms of symbol error rate, some conventional coherent pilot-based detectors, including a sphere decoder and a joint channel estimation--data detection scheme.} Our EP-based detectors produce accurate approximates of the true posterior leading to high achievable sum-rates. The gains of these detectors are further observed in terms of the bit error rate when using their soft outputs for a turbo channel decoder. 
\end{abstract}

\section{Introduction} \label{sec:intro}
In wireless communications, multi-antenna based multiple-input multiple-output (MIMO) technology is capable of improving significantly both the system spectral efficiency and reliability due to its multiplexing and diversity gains~\cite{Telatar1999capacityMIMO,Foschini}. MIMO is at the heart of current cellular systems, and large-scale (massive) MIMO~\cite{Emil2017_massivemimobook} 
is considered as one of the fundamental technologies for the fifth-generation~(5G) wireless communications~\cite{Larsson2017massiveMIMO_5G}. 
In practical MIMO systems, the transmitted symbols are normally drawn from a finite discrete constellation to reduce complexity. Due to propagation effects, the symbols sent from different transmit antennas interfere, and the receiver observes a linear superposition of these symbols corrupted by noise. 
The task of the receiver is to detect these symbols (or rather the underlying bits) based on the received signal and the available knowledge about the channel. 

If the {\em instantaneous} value of the channel matrix is treated as known, the detection is said to be {\em coherent} and has been investigated extensively in the literature~\cite{Hanzo2015_50years}. In this case, the data symbols are normally taken from a scalar constellation such as the quadrature amplitude modulation~(QAM). 
Since the optimal maximum-likelihood~(ML) coherent detection problem is known to be non-deterministic polynomial-time hard (NP-hard) 
\cite{Verdu1989computational}, 
many sub-optimal coherent MIMO detection algorithms have been proposed. These range from linear schemes, such as the zero forcing~(ZF) and minimum mean square error~(MMSE) detectors, to non-linear schemes based on, for example, interference cancellation, tree search, and lattice reduction~\cite{Hanzo2015_50years}.  

If only {\em statistical} information about the channel is available, the detection problem is said to be {\em non-coherent}. \revise{In the block fading channel where the channel matrix remains constant for each length-$T$ coherence block and varies between blocks, the receiver can estimate (normally imperfectly) the channel based on the transmitted pilot symbols, then perform coherent detection based on the channel estimate. Channel estimation and data detection can also be done iteratively~\cite{Buzzi2004_iterative_dataDetection-channelEstimation, Hanzo2008semiblind}, or jointly based on tree search~\cite{Xu2008_exactML,Alshamary2016efficient}. These schemes requires pilot transmission for an initial channel estimate or to guarantee the identifiability of the data symbols.  
Another approach not involving pilot transmission is unitary space time modulation, in which the matrix of symbols in the space-time domain is orthonormal and isotropically distributed~\cite{Hochwald2000unitaryspacetime}.} 
There, information is carried by the signal matrix subspace position, which is invariant to multiplication by the channel matrix. Therefore, a constellation over matrix-valued symbols can be designed as a collection of subspaces in $\CC^T$. Such constellations belong to the Grassmann manifold $G(\CC^T,K)$, which is the space of $K$-dimensional subspaces in $\CC^T$, where $K$ is the number of transmit antennas. For the independent and identically distributed~(i.i.d.) Rayleigh block fading channel, when the signal-to-noise-ratio~(SNR) is large, Grassmannian signaling was shown to achieve a rate within a vanishing gap from the capacity if $T\ge N + \min\{K,N\}$~\cite{ZhengTse2002Grassman}, and within a constant gap if $2K\le T \le N + K$~\cite{Yang2013CapacityLargeMIMO}, where $N$ is the number of receive antennas. 
Like with coherent detection, the optimal ML non-coherent detection problem under Grassmannian signaling is NP-hard. Thus, low-complexity sub-optimal detectors have been proposed for constellations with additional structure, e.g.,~\cite{Hochwald2000systematicDesignUSTM,Kammoun2007noncoherentCodes,Hoang_cubesplit_journal}. 

In this paper, we focus on the non-coherent detection problem in the Rayleigh flat and block fading single-input multiple-output~(SIMO) multiple-access channel~(MAC) with coherence time $T$. There, the communication signals are independently transmitted from $K$ single-antenna users. If the users could cooperate, the high-SNR optimal joint signaling scheme would be a Grassmannian signaling on $G(\CC^T,K)$~\cite{ZhengTse2002Grassman}. However, we assume uncoordinated users, for which the optimal non-coherent transmission scheme is not known, \revise{although some approximate optimality design criteria have been proposed in~\cite{Hoang2020const_MAC}. 
In this work, we design the detector without assuming any specific structure of the signal transmitted over a coherence block.}
We consider the case where the receiver is interested not only in the hard detection of the symbols but also in their posterior marginal probability mass functions~(\revise{PMFs}). This ``soft'' information is needed, for example, when computing the bit-wise log-likelihood ratios~(LLRs) required for soft-input soft-output channel decoding. Computing an exact marginal \revise{PMF} would require enumerating all possible combinations of other-user signals, which is infeasible with many users, many antennas, or large constellations. Thus, we seek sub-optimal schemes with practical complexity. 

In contrast to probabilistic coherent MIMO detection, for which many schemes have been proposed (e.g.,~\cite{Caire2004iterative,Goldberger2011GTA,Cespedes2018MIMOwithEC}), the probabilistic non-coherent MIMO detection under \revise{general signaling, and Grassmannian signaling in particular,} has not been well investigated. The detection scheme proposed in~\cite{HoangAsilomar2018multipleAccess} is sub-optimal and compatible only with the specific constellation structure considered therein. The list-based soft demapper in \cite{Davidson2009BICM_IDD} reduces the number of terms considered in posterior marginalization by including only those symbols at a certain distance from a reference point. However, it was designed for the single-user case only and has no obvious generalization to the MAC. 
\revise{The semi-blind approaches~\cite{Buzzi2004_iterative_dataDetection-channelEstimation, Hanzo2008semiblind,Xu2008_exactML,Alshamary2016efficient} 
for the MIMO point-to-point channel can be extended to the MAC. However, these schemes are restricted to transmitted signals with pilots.}

In this work, we propose message-passing algorithms for posterior marginal inference of non-coherent multi-user MIMO transmissions over \revise{spatially correlated} Rayleigh block fading channels. Our algorithms are based on expectation propagation~(EP) approximate inference~\cite{Minka:Diss:01,Heskes:JSM:05}. EP provides an iterative framework for approximating posterior beliefs by parametric distributions in the exponential family~\cite[Sec.~1.6]{Doksum1977mathStatistics_ideas_and_topics}. Although there are many possible ways to apply EP to our non-coherent multi-user detection problem, we do so by \revise{choosing as variable nodes} the indices of the transmitted symbols and the noiseless received signal from each user. 
The EP algorithm passes messages between the corresponding variable nodes and factor nodes on a bipartite factor graph. In doing so, the approximate posteriors of these variables are iteratively refined. 
We also address numerical implementation issues.

To measure the accuracy of the approximate posterior generated by the soft detectors, we compute the mismatched sum-rate of the system that uses the approximate posterior as the decoding metric. This mismatched sum-rate approaches the achievable rate of the system as the approximate posterior gets close to the true posterior. 
We also evaluate the symbol error rate when using the proposed schemes for hard detection, and the bit error rate when using these schemes for turbo equalization with a standard turbo code.

The contributions of this work are summarized as follows:
\begin{enumerate}[leftmargin=*]
	\item We propose soft and hard multi-user detectors for the non-coherent SIMO MAC using EP approximate inference, and methods to stabilize the EP updates. \revise{The proposed detectors work for general vector-valued transmitted symbols within each channel coherence block, i.e., it is general enough to include both the pilot-assisted and pilot-free signaling cases.}
	
	\item We propose two simplifications of the EP detector with reduced complexity. 
	The first one, so-called EPAK, is based on approximating the EP messages with Kronecker products. The second one can be interpreted as soft MMSE estimation and successive interference \revise{approximation}~(\revise{SIA}).
	
	\item We analyze the complexity and numerically evaluate the \revise{convergence, running time}, and performance of the proposed EP, EPAK, and MMSE-\revise{SIA} detectors, the optimal ML detector, a genie-aided detector, the state-of-the-art detector from~\cite{HoangAsilomar2018multipleAccess}, and \revise{some conventional coherent pilot-based schemes}. 
	Our results suggest that the proposed detectors offer significantly improved mismatched sum-rate, symbol error rate, and coded bit error rate with respect to (w.r.t.) some existing sub-optimal schemes, while having lower complexity than the ML detector. 
\end{enumerate} 
To the best of our knowledge, our proposed approach is the first message-passing scheme for non-coherent multi-user MIMO detection \revise{with general constellations}.

The remainder of this paper is organized as follows. The system model is presented in Section~\ref{sec:model}.  A brief review of EP is presented in Section~\ref{sec:EP}, and the EP approach to non-coherent detection is presented in Section~\ref{sec:EP_application}. In Section~\ref{sec:MMSE-SIA_EPAK}, two simplifications (MMSE-\revise{SIA} and EPAK) of the EP detector are presented. Implementation aspects of EP, MMSE-\revise{SIA}, and EPAK are discussed in Section~\ref{sec:implementation}. Numerical results and conclusions are presented in Section~\ref{sec:performance} and Section~\ref{sec:conclusion}, respectively. The mathematical preliminaries and proofs are provided in the appendices.

{\it Notations:}
Random quantities are denoted with non-italic letters with sans-serif fonts, e.g., a scalar $\rv{x}$, a vector $\rvVec{v}$, and a matrix $\rvMat{M}$. Deterministic quantities are denoted 
with italic letters, e.g., a scalar $x$, a vector $\pmb{v}$, and a
matrix $\pmb{M}$. 
The Euclidean norm is denoted by $\|\vv\|$ and the Frobenius norm $\|\Mm\|_F$. 
\revise{The conjugate, 
transpose, conjugate transpose, trace, and vectorization of $\Mm$ are denoted by $\Mm^*$, 
$\Mm^\T$, $\Mm^\H$, $\trace\{\Mm\}$, and $\vectr{\Mm}$, respectively.} 
$\prod$ denotes the conventional or Cartesian product, depending on the factors. $\otimes$ denotes the Kronecker product.
$\mathbbm{1}\{A\}$ denotes the indicator function whose value is 1 if $A$ is true and 0 otherwise. 
$[n] \defeq \{1,2,\dots,n\}$. $\propto$ means ``proportional to''. 
The Grassmann manifold $G(\mathbb{C}^T,K)$ is the space of $K$-dimensional subspaces in $\mathbb{C}^T$. 
In particular, $G(\mathbb{C}^T,1)$ is the Grassmannian of lines. 
The Kullback-Leibler divergence of a distribution $p$ from another distribution $q$ of a random vector $\rvVec{x}$ with domain $\Xc$ is defined by $D(q\|p) \defeq \int_{\Xc} q(\xv) \log \frac{q(\xv)}{p(\xv)} \dif \xv$ if $\Xc$ is continuous and $D(q\|p) \defeq \sum_{\xv\in \Xc} q(\xv) \log \frac{q(\xv)}{p(\xv)}$ if $\Xc$ is discrete. 
$\Nc(\muv,\Sigmam)$ denotes the {\em complex} Gaussian vector distribution with mean $\muv$, covariance matrix $\Sigmam$, and thus probability density function~(PDF)
\begin{align}
\Nc(\xv; \muv,\Sigmam) \defeq \frac{\exp\big(-(\xv-\muv)^\H \Sigmam^{-1}(\xv-\muv)\big)}{\pi^n \det(\Sigmam)} , ~ \xv \in \CC^n.
\end{align}

\section{System Model} \label{sec:model}
\subsection{Channel Model} 
We consider a SIMO MAC in which $K$ single-antenna users transmit to an $N$-antenna receiver. We assume that the channel is flat and block fading with an equal-length and synchronous (across the users) coherence interval of $T$ channel uses. That is, the channel vectors $\rvVec{h}_k \in \CC^{N\times 1}$, which contain the fading coefficients between the transmit antenna of user $k\in [K]$ and the $N$ receive antennas, remain constant within each coherence block of $T$ channel uses and change independently between blocks. 
Furthermore, the {\em distribution} of $\rvVec{h}_k$ is assumed to be known to the receiver, but its {\em realizations} are  unknown to both ends of the channel. 
\revise{Since the users are not co-located, we assume that the $\rvVec{h}_k$ are independent across users. We consider Rayleigh fading with receiver-side correlation, i.e., $\rvVec{h}_k \sim \Nc(\mathbf{0},\Xim_k)$, where $\Xim_k \in \CC^{N\times N}$ is the spatial correlation matrix. We assume that $\frac{1}{N}\trace[\Xim_k] \eqdef \xi_k$ where $\xi_k$ is the large-scale average channel gain from one of the receive antennas to user $k$}. 
We assume that $T >  K$ and $N \ge K$. 

Within a coherence block, each transmitter $k$ sends a signal vector $\rvVec{s}_k \in \CC^T$, and the receiver receives a realization $\Ym$ of the random matrix
\begin{align}
\rvMat{Y} &= \sum_{k=1}^{K} \rvVec{s}_k \rvVec{h}_k^\T + \rvMat{W} = \rvMat{S}\rvMat{H}^\T + \rvMat{W}, \label{eq:channel_model}
\end{align}
where $\rvMat{S} = [\rvVec{s}_1 \ \dots \ \rvVec{s}_K] \in \CC^{T\times K}$  and $\rvMat{H} = [\rvVec{h}_1 \ \dots \ \rvVec{h}_K] \in \CC^{N\times K}$ concatenate the transmitted signals and channel vectors, respectively, $\rvMat{W} \in \CC^{T\times N}$ is the Gaussian noise with i.i.d.~$\Nc(0,\sigma^2)$ entries independent of $\rvMat{H}$, and  the block index is omitted for simplicity. 

We assume that the transmitted signals have average unit norm, i.e., $\E[\|\rvVec{s}_k\|^2] = 1, k \in [K]$.  Under this normalization, the signal-to-noise ratio~(SNR) of the transmitted signal from user $k$ at each receive antenna is $\SNR_k = \revise{\xi_k}/(T\sigma^2)$. 
We assume that the transmitted signals belong to {\em disjoint} finite discrete individual \revise{constellations with vector-valued symbols. That is, 
$
\rvVec{s}_k \in \Sc_k \defeq \{\sv_k\of{1},\dots,\sv_{k}\of{|\Sc_k|}\}, ~ k \in [K].
$
In particular, $\Sc_k$ can be a Grassmannian constellation on $G(\CC^T,1)$, i.e., each constellation symbol $\sv_k\of{i}$ is a unit-norm vector representative of a point in $G(\CC^T,1)$. Another example is when the constellation symbols contain pilots and scalar data symbols.\footnote{\revise{In this case, the constellations are disjoint thanks to the fact that pilot sequences are user-specific.}}
 Each symbol in $\Sc_k$ is labeled with a binary sequence of length $B_k \defeq \log_2|\Sc_k|$.} 

\subsection{Multi-User Detection Problem}
\revise{Given $\rvMat{S} = \Sm = [\sv_1, \ \sv_2, \ \dots, \ \sv_K]$, 
the conditional probability density $p_{\rvMat{Y} | \rvMat{S}}$, also known as likelihood function, is derived similar to~\cite[Eq.(9)]{Jafar2005correlatedFading} as
\begin{multline}
p_{\rvMat{Y} | \rvMat{S}}(\Ym | \Sm) = \\ \frac{\exp\big(-\vectr{\Ym^\T}^\H\big(\sigma^2\Id_{NT}+\sum_{k=1}^{K} \sv_k\sv_k^\H \otimes \Xim_k \big)^{-1}\vectr{\Ym^\T} \big)}{\pi^{NT}\det(\sigma^2\Id_{NT}+\sum_{k=1}^{K}\sv_k\sv_k^\H \otimes \Xim_k)}.
\label{eq:likelihood}
\end{multline}
}
Given the received signal $\rvMat{Y} = \Ym$, the joint multi-user ML symbol decoder is then 
\revise{
\begin{align} 
\hat{\Sm} 
&= \arg \! \min_{\Sm \in \prod_{k=1}^{K}\Sc_k} 
\!\bigg(\!\vectr{\Ym^\T}^\H\Big(\sigma^2\Id_{NT}\!+\!\sum_{k=1}^{K}\sv_k\sv_k^\H \otimes \Xim_k\Big)^{-1}\!\!\vectr{\Ym^\T} \bigg. \notag \\
&\bigg.\qquad+ \log\det\Big(\sigma^2\Id_{NT}+\sum_{k=1}^{K}\sv_k\sv_k^\H \otimes \Xim_k\Big)\bigg). \label{eq:MLdecoder} 
\end{align}
Since the ML decoding metric depends on $\Sm$ only through $\sum_{k=1}^{K}\sv_k\sv_k^\H \otimes \Xim_k$, for identifiability, it must hold that $\sum_{k=1}^{K}\sv_k\sv_k^\H \otimes \Xim_k \ne \sum_{k=1}^{K}\sv'_k{\sv'_k}^\H \otimes \Xim_k$ for any pair of distinct joint symbols $\Sm = [\sv_1, \dots, \sv_K]$ and $\Sm'= [\sv'_1, \dots, \sv'_K]$ in $\prod_{k=1}^{K}\Sc_k$. 
}

When a channel code is used, most channel decoders require the LLRs of the bits. \revise{The LLR of the $j$-th bit of user $k$, denoted by $\rv{b}_{k,j}$, given the observation $\rvMat{Y}=\Ym$ is defined as 
\begin{align}
\LLR_{k,j}(\Ym) &\defeq \log \frac{p_{\rvMat{Y}|\rv{b}_{k,j}}(\Ym|1)}{p_{\rvMat{Y}|\rv{b}_{k,j}}(\Ym|0)} 
\\
&= \log \frac{\sum_{\alphav \in \Sc_{k,j}^{(1)}}p_{\rvMat{Y}|\rvVec{s}_k}(\Ym|\alphav)}{\sum_{\betav \in \Sc_{k,j}^{(0)}}p_{\rvMat{Y}|\rvVec{s}_k}(\Ym|\betav)} \\
&= \log \frac{\sum_{\alphav \in \Sc_{k,j}^{(1)}}p_{\rvVec{s}_k| \rvMat{Y}}(\alphav|\Ym)}{\sum_{\betav \in \Sc_{k,j}^{(0)}}p_{\rvVec{s}_k|\rvMat{Y}}(\betav | \Ym)} 
\label{eq:LLR}
\end{align}
where $\Sc_{k,j}^{(b)}$ denotes the set of all possible symbols in $\Sc_k$ with the $j$-th bit being equal to $b$ 
for $j \in [B_k]$ and $b \in \{0,1\}$. To compute \eqref{eq:LLR}, the posteriors $p_{\rvVec{s}_k|\rvMat{Y}}$, $k\in [K]$, are marginalized from}
\begin{align}
p_{\rvMat{S} | \rvMat{Y}}(\Sm | \Ym) = \frac{p_{\rvMat{Y}|\rvMat{S}}(\Ym|\Sm) p_{\rvMat{S}}(\Sm)}{p_{\rvMat{Y}}(\Ym)} \propto p_{\rvMat{Y}|\rvMat{S}}(\Ym|\Sm) p_{\rvMat{S}}(\Sm). \label{eq:posterior}
\end{align}
Assuming that the transmitted signals are independent and uniformly distributed over the respective constellations, the prior $p_{\rvMat{S}}$ factorizes as 
\begin{align}
\Pr({\rvMat{S}} = [\sv_1,\dots,\sv_K]) = 
 \prod_{k=1}^{K}\frac{1}{|\Sc_k|} \mathbbm{1}\{\sv_k\in \Sc_k\}.
\end{align}
On the other hand, the likelihood function $p_{\rvMat{Y}|\rvMat{S}}(\Ym|[\sv_1,\dots,\sv_K])$ involves all $\sv_1,\dots,\sv_K$ in such a manner that it does not straightforwardly factorize. Exact marginalization of $p_{\rvMat{S}|\rvMat{Y}}$ requires computing
\begin{align} \label{eq:exact_marginalization}
p_{\rvVec{s}_k| \rvMat{Y}}(\sv_k|\Ym) = \sum_{\sv_l \in \Sc_l, \forall l\ne k} p_{\rvMat{S}|\rvMat{Y}}([\sv_1,\dots,\sv_K] | \Ym), \quad k \in [K]. \quad
\end{align}
That is, it requires computing $p_{\rvMat{Y} | \rvMat{S}}(\Ym | \Sm)$ (which requires the inversion of an \revise{$NT\times NT$} matrix) for all $\Sm \in \prod_{k=1}^{K}\Sc_k$. Thus, the total complexity of exact marginalization is $O(\revise{K^6} 2^{KB})$.\footnote{Throughout the paper, as far as the complexity analysis is concerned, we assume for notational simplicity that  $T = O(K)$, $N = O(K)$, and $|\Sc_k| = O(2^B), ~\forall k\in[K]$. \revise{If the channels are uncorrelated ($\Xim_k = \Id_N$), the likelihood function can be simplified as $p_{\rvMat{Y} | \rvMat{S}}(\Ym | \Sm) = \frac{\exp\left(-\trace\left\{\Ym^\H(\sigma^2\Id_T+\Sm\Sm^\H)^{-1}\Ym \right\}\right)}{\pi^{NT}\det^N(\sigma^2\Id_T+\Sm\Sm^\H)}$. Thus, the complexity of exact marginalization is reduced to $O(K^3 2^{KB})$.}}
This is formidable for many users or large constellations. Thus, we seek alternative approaches to estimate
\begin{align}
p_{\rvMat{S}|\rvMat{Y}}([\sv_1,\dots,\sv_K] | \Ym) &\approx \hat{p}_{\rvMat{S}|\rvMat{Y}}([\sv_1,\dots,\sv_K]| \Ym) \\
&= \prod_{k=1}^{K} \hat{p}_{\rvVec{s}_k | \rvMat{Y}}(\sv_k | \Ym). \label{eq:post_approx}
\end{align} 

\subsection{Achievable Rate}
According to \cite[Sec.~II]{Ganti2000mismatchedDecoding}, the highest sum-rate reliably achievable with a given decoding metric $\hat{p}_{\rvMat{S}|\rvMat{Y}}$, so-called the mismatched sum-rate, is lower bounded by the generalized mutual information~(GMI) given by
\begin{align}
	&R_{\rm GMI} \notag \\&= \frac{1}{T}\sup_{s\ge 0} \E\bigg[\log_2 \frac{{\hat{p}_{\rvMat{S}|\rvMat{Y}}}(\rvMat{S}|\rvMat{Y})^s}{\sum_{\Sm' \in \prod_{k=1}^K \Sc_k} \Pr({\rvMat{S}} = \Sm') {\hat{p}_{\rvMat{S}|\rvMat{Y}}(\Sm' | \rvMat{Y})}^s}\bigg] 
\\
&= \frac{1}{T}\sup_{s\ge 0} \E\bigg[\sum_{k=1}^{K}B_k - \log_2\frac{\sum_{\Sm' \in \prod_{k=1}^K \Sc_k}  {\hat{p}_{\rvMat{S}|\rvMat{Y}}(\Sm'|\Ym)}^s}{{\hat{p}_{\rvMat{S}|\rvMat{Y}}(\rvMat{S}|\rvMat{Y})}^s}\bigg] \label{eq:tmp407}\\
&= \frac{1}{T}\sum_{k=1}^{K}B_k - \frac{1}{T}\inf_{s\ge 0} \E\bigg[\sum_{k=1}^{K}\log_2 \frac{\sum_{\sv'_k \in \Sc_k}{\hat{p}_{\rvVec{s}_k|\rvMat{Y}}(\sv'_k | \rvMat{Y})}^s}{{\hat{p}_{\rvVec{s}_k|\rvMat{Y}}(\rvVec{s}_k|\rvMat{Y})}^s}\bigg] \quad \label{eq:approx_achievableRate}
\end{align}
bits/channel use,
where the expectation is over the joint distribution of $\rvMat{S}$ and $\rvMat{Y}$, i.e., $p_{\rvMat{Y}|\rvMat{S}}p_{\rvMat{S}}$, \eqref{eq:tmp407} holds because the transmitted symbols are independent and have uniform prior distribution, and \eqref{eq:approx_achievableRate} follows from the factorization of $\hat{p}_{\rvMat{S}|\rvMat{Y}}$ in \eqref{eq:post_approx}. 
The generalized mutual information $\GMI$ is upper bounded by the sum-rate achieved with the optimal decoding metric ${p}_{\rvMat{S}|\rvMat{Y}}$ given by
\begin{align}
R &= \frac{1}{T} I(\rvMat{S};\rvMat{Y}) 
\\
&= \frac{1}{T}h(\rvMat{S}) - \frac{1}{T}h(\rvMat{S}| \rvMat{Y}) \\
&= \frac{1}{T}\sum_{k=1}^{K}B_k-\frac{1}{T}\E\bigg[\log_2\frac{1}{p_{\rvMat{S}|\rvMat{Y}}(\rvMat{S}|\rvMat{Y})}\bigg] \label{eq:tmp446} \\
&= \frac{1}{T}\sum_{k=1}^{K}B_k-\frac{1}{T}\E\bigg[\log_2\frac{\sum_{\Sm' \in \prod_{k=1}^K \Sc_k} p_{\rvMat{Y}|\rvMat{S}}(\rvMat{Y}|\Sm')}{p_{\rvMat{Y}|\rvMat{S}}(\rvMat{Y}|\rvMat{S})}\bigg] \label{eq:achievableRate2}
\end{align}
bits/channel use, where \eqref{eq:achievableRate2} follows from the Bayes' law and the uniformity of the prior distribution.
$\GMI$ approaches $R$ as $\hat{p}_{\rvMat{S}|\rvMat{Y}}$ gets close to ${p}_{\rvMat{S}|\rvMat{Y}}$. Note that 
if we fix $s = 1$ \revise{in place of the infimum in~\eqref{eq:approx_achievableRate}}, it holds that 
\begin{align}
R - \GMI(s = 1) &= \frac{1}{T}\E[\log_2\frac{p_{\rvMat{S}|\rvMat{Y}}(\rvMat{S}|\rvMat{Y})}{\hat{p}_{\rvMat{S}|\rvMat{Y}}(\rvMat{S}|\rvMat{Y})}] \\
&= \frac{1}{T} \E_{\rvMat{Y}}\big[D({p_{\rvMat{S}|\rvMat{Y}}} \|{\hat{p}_{\rvMat{S}|\rvMat{Y}}})\big],
\end{align}
which converges to zero when the KL divergence between $\hat{p}_{\rvMat{S}|\rvMat{Y}}$ and $p_{\rvMat{S}|\rvMat{Y}}$ vanishes.

The expectations in~\eqref{eq:approx_achievableRate} and~\eqref{eq:achievableRate2} cannot be derived in closed form in general. Alternatively, we can evaluate $R$ and $\GMI$ (and also  $\E_{\rvMat{Y}}[D({p_{\rvMat{S}|\rvMat{Y}}}\big\|{\hat{p}_{\rvMat{S}|\rvMat{Y}}})]$) numerically with the Monte Carlo method. 
Note that when $K$ or $B_k$ is large, even a numerical evaluation of $R$ and $\E_{\rvMat{Y}}[D({p_{\rvMat{S}|\rvMat{Y}}}\big\|{\hat{p}_{\rvMat{S}|\rvMat{Y}}})]$ is not possible. Therefore, we choose to use the mismatched sum-rate lower bound $\GMI$ as an information-theoretic metric to evaluate how close  $\hat{p}_{\rvMat{S}|\rvMat{Y}}$ is to $p_{\rvMat{S}|\rvMat{Y}}$.

In what follows, we design a posterior marginal estimation scheme based on EP. We start by providing a brief review of EP in the next section.

\section{Expectation Propagation} \label{sec:EP}
The EP algorithm was first proposed in \cite{Minka:Diss:01} and summarized in, e.g., \cite{Heskes:JSM:05} for approximate inference in probabilistic graphical models. EP is an iterative framework for approximating posterior beliefs by parametric distributions in the exponential family~\cite[Sec.~1.6]{Doksum1977mathStatistics_ideas_and_topics}.
Let us consider a set of unknown variables represented by a random vector $\rvVec{x}$ with posterior of the form
\begin{align}
p_{\rvMat{x}}(\xv)
&\propto \prod_{\alpha} \psi_{\alpha}(\xv_\alpha) 
\label{eq:p_alf},
\end{align}
where $\xv_\alpha$ is the subset of variables involved in the factor $\psi_{\alpha}$ corresponding to a partition $\{\rvVec{x}_\alpha\}$ of $\rvVec{x}$.  
Furthermore, let us partition the components of $\rvVec{x}$ into some sets $\{\rvVec{x}_\beta\}$, where no $\rvVec{x}_\beta$ is split across factors (i.e., $\forall \alpha,\beta$ either $\rvVec{x}_\beta\subset\rvVec{x}_\alpha$ or $\rvVec{x}_\beta\cap\rvVec{x}_\alpha=\emptyset$). \revise{The partition $\{\rvVec{x}_\alpha\}$ represents the local dependency of the variables given by the intrinsic factorization \eqref{eq:p_alf}, while the partition $\{\rvVec{x}_\beta\}$ groups the variables that always occur together in a factor.} We are interested in the posterior marginals w.r.t. the partition $\{\rvVec{x}_\beta\}$. In the following, we omit $\rvVec{x}$ in the subscripts since it is obvious.

EP approximates the true posterior $p$ from \eqref{eq:p_alf} by a distribution $\hat{p}$ that can be expressed in two ways. First, it can be expressed w.r.t. the ``target'' partition $\{\rvVec{x}_\beta\}$ as
\begin{align}
\hat{p}(\xv)
&= \prod_{\beta} \hat{p}_\beta(\xv_\beta) 
\label{eq:ptil_beta} ,
\end{align}
where $\hat{p}_\beta$ are constrained to be in the exponential family~\cite[Sec.~1.6]{Doksum1977mathStatistics_ideas_and_topics}, such that~(s.t.)
\begin{align}
\hat{p}_\beta(\xv_\beta) 
&= 
\exp\big(\gammav_\beta^\T \phiv_\beta(\xv_\beta) - A_\beta(\gammav_\beta)\big) 
\label{eq:expfam},
\end{align}
for 
sufficient statistics $\phiv_\beta(\xv_\beta)$, 
parameters $\gammav_\beta$, 
and log-partition function 
$A_\beta(\gammav) \defeq \ln \int 
e^{\gammav^\T\phiv_\beta(\xv_\beta)} \dif\xv_\beta$.
Second, $\hat{p}$ can also be expressed w.r.t. the partition $\{\rvVec{x}_\alpha\}$ %
as
\begin{align}
\hat{p}(\xv)
&\propto \prod_{\alpha} m_\alpha(\xv_\alpha) 
\label{eq:ptil_alf} ,
\end{align}
in accordance with~\eqref{eq:p_alf}. 
For \eqref{eq:ptil_beta} and \eqref{eq:ptil_alf} to be consistent, the terms $m_\alpha$ should also factorize over $\beta$, i.e., there exist factors $\m{\alpha}{\beta}$ of the form
$
\m{\alpha}{\beta}(\xv_\beta) 
= 
\exp\big(\gam{\alpha}{\beta}^\T \phiv_\beta(\xv_\beta)
\big) 
$
s.t.
\begin{align}
\!\!\!\!\!\!\!\!\!\!m_\alpha(\xv_\alpha)
&= \prod_{\beta \in \Nalpha}  \m{\alpha}{\beta}(\xv_\beta) 
= \exp\bigg( \sum_{\beta\in\Nalpha}\gam{\alpha}{\beta}^\T \phiv_\beta(\xv_\beta) \bigg),
\\
\!\!\!\!\!\!\!\!\!\!\hat{p}_\beta(\xv_\beta)
&\propto \prod_{\alpha \in \Nbeta}  \m{\alpha}{\beta}(\xv_\beta) 
= \exp\bigg( \sum_{\alpha\in\Nbeta}\gam{\alpha}{\beta}^\T \phiv_\beta(\xv_\beta) \bigg), \label{eq:phat_beta}
\end{align}
where $\Nalpha$ collects the indices $\beta$ for which $\rvVec{x}_\beta\subset \rvVec{x}_\alpha$,
and $\Nbeta$ collects the indices $\alpha$ for which $\rvVec{x}_\beta\subset \rvVec{x}_\alpha$.
It turns out that $\m{\alpha}{\beta}$ can be interpreted as a message from the factor node $\alpha$ to the variable node $\beta$ on a bipartite factor graph~\cite{Kschischang2001factor_graph}.
In this case, 
$\hat{p}_\beta(\xv_\beta)$ is proportional to the product of all messages impinging on variable node $\beta$. 

EP works by first initializing all $m_\alpha(\xv_\alpha)$ 
and $\hat{p}_\beta(\xv_\beta)$ (typically by the respective priors, which are assumed to also belong to the considered exponential family), then iteratively updating each approximation factor $m_\alpha$ in turn. Let us fix a factor index $\alpha$. According to~\cite{Minka:Diss:01}, the ``tilted'' distribution $q_\alpha$ is constructed by swapping the true potential $\psi_\alpha$ for its approximate $m_\alpha$ in $\hat{p}(\xv)$ as
$
q_{\alpha}(\xv)
= \frac{\hat{p}(\xv) \psi_\alpha(\xv_\alpha)}
{m_\alpha(\xv_\alpha)} ,
$
where it is assumed that $\int q_\alpha(\xv) \dif \xv < \infty$. This tilted distribution is projected back onto the exponential family by minimizing the KL divergence:
\begin{align}
\hat{p}\new_\alpha(\xv)
&= \arg\min_{\underline{p}\in\Pc} D\big( q_{\alpha}(\xv) \big\| \underline{p}(\xv )\big) 
\label{eq:pnew},
\end{align}
where $\Pc$ is the set of distributions of the form of $\hat{p}$ in \eqref{eq:ptil_beta}, i.e., 
$
\underline{p}(\xv)
= \prod_{\beta} \underline{p}_\beta(\xv_\beta) = \prod_{\beta}\exp\big(\underline{\gammav}_\beta^\T \phiv_\beta(\xv_\beta) - A_\beta(\underline{\gammav}_\beta)\big)
$
for some $\{\underline{\gammav}_\beta\}$. Following~\cite{Minka:Diss:01}, 
the solution to \eqref{eq:pnew} is as follows.
\begin{proposition} \label{prop:pnew}
	The solution to \eqref{eq:pnew} is given by $\hat{p}\new_\alpha(\xv) = \prod_\beta \hat{p}\new_{\alpha,\beta}(\xv_\beta)$ with 
	$
	\hat{p}\new_{\alpha,\beta}(\xv_\beta) = 
	{\hat{p}}_\beta(\xv_\beta)$,
	$\forall \beta\notin\Nalpha$, and $\hat{p}\new_{\alpha,\beta}(\xv_\beta) = 	\exp\big(\uvec{\gamma}_\beta^\T \phiv_\beta(\xv_\beta)-A_\beta(\uvec{\gamma}_\beta)\big)$ with $\uvec{\gamma}_\beta$ s.t. $\E_{\hat{p}\new_{\alpha,\beta}}[\phiv_\beta(\xv_\beta)]
	=\E_{q_\alpha}[\phiv_\beta(\xv_\beta)]$,
	$\forall \beta\in\Nalpha$, whenever the expectation $\E_{q_\alpha}[\cdot]$ exists.
\end{proposition}
\begin{proof}
	The proof is given in Appendix~\ref{proof:pnew}.
\end{proof}
The factor $m_\alpha$ is then updated via
\begin{align}
	m_\alpha\new(\xv_\alpha)
&= \frac{\hat{p}\new_\alpha(\xv) m_{\alpha}(\xv_\alpha)}
{\hat{p}(\xv)} 
\label{eq:malf1} \\
&= \bigg[\prod_{\beta\in\Nalpha} \m{\alpha}{\beta}(\xv_\beta)\bigg] 
\frac{\prod_{\beta\in\Nalpha} \hat{p}\new_{\alpha,\beta}(\xv_\beta)}
{\prod_{\beta\in\Nalpha} \hat{p}_\beta(\xv_\beta)} \\
&\propto \bigg[\prod_{\beta\in\Nalpha} \m{\alpha}{\beta}(\xv_\beta)\bigg] 
\notag \\&\quad \times 
\frac{\prod_{\beta\in\Nalpha} \hat{p}\new_{\alpha,\beta}(\xv_\beta)}
{\prod_{\beta\in\Nalpha} \big[ 
	\m{\alpha}{\beta}(\xv_\beta)
	\prod_{\alpha'\in\Nbeta\setminus\alpha}
	\m{\alpha'}{\beta}(\xv_\beta)
	\big]} \\
&= \prod_{\beta\in\Nalpha} 
\m{\alpha}{\beta}\new(\xv_\beta)
\label{eq:malf2} ,
\end{align}
with 
\begin{align}
\m{\alpha}{\beta}\new(\xv_\beta) \defeq \frac{\hat{p}\new_{\alpha,\beta}(\xv_\beta)}
{ \prod_{\alpha'\in\Nbeta\setminus\alpha}
	\m{\alpha'}{\beta}(\xv_\beta)}.\label{eq:malf_new}
\end{align} 
Note that, on the right-hand side~(RHS) of \eqref{eq:malf1}, all terms dependent on $\{\xv_\beta\}_{\beta\notin\Nalpha}$ cancel, leaving the dependence only on $\{\xv_\beta\}_{\beta\in\Nalpha}$. Thus, the update of $m_\alpha$ only affects the approximate posterior of nodes $\beta$ in the neighborhood of node $\alpha$. After that, the process is repeated with the next $\alpha$.

A message-passing view of Proposition~\ref{prop:pnew} can be seen by expanding $q_\alpha(\xv)$ as
\begin{align}
	q_{\alpha}(\xv)
	&= \frac{\psi_\alpha(\xv_\alpha)}{m_\alpha(\xv_\alpha)}
	\bigg[\prod_{\beta\in\Nalpha} \prod_{\alpha'\in\Nbeta}\m{\alpha'}{\beta}(\xv_\beta) \bigg] 
	\bigg[\prod_{\beta\notin\Nalpha} \hat{p}_\beta(\xv_\beta) \bigg] 
\\
&= \psi_\alpha(\xv_\alpha) 
\bigg[\prod_{\beta\in\Nalpha} \prod_{\alpha'\in\Nbeta\setminus\alpha}   \m{\alpha'}{\beta}(\xv_\beta) \bigg] 
\bigg[\prod_{\beta\notin\Nalpha} \hat{p}_\beta(\xv_\beta) \bigg] 
\label{eq:qalf2},
\end{align}
then, using the natural logarithm for the KL divergence, 
it follows that
\begin{align}
\lefteqn{
	D\big( q_{\alpha}(\xv) \big\| \underline{p}(\xv)\big) 
}\nonumber \\
&= \int q_{\alpha}(\xv) \ln \frac{q_{\alpha}(\xv)}{ \underline{p}(\xv )} \dif \xv \\
&= \int 
\psi_\alpha(\xv_\alpha) 
\bigg[\prod_{\beta\in\Nalpha} \prod_{\alpha'\in\Nbeta\setminus\alpha} \m{\alpha'}{\beta}(\xv_\beta) \bigg] 
\bigg[\prod_{\beta\notin\Nalpha} \hat{p}_\beta(\xv_\beta) \bigg] 
\nonumber\\&\quad\times
\ln \bigg(\frac{ 
	\psi_\alpha(\xv_\alpha) 
	\prod_{\beta\in\Nalpha} \prod_{\alpha'\in\Nbeta\setminus\alpha} \m{\alpha'}{\beta}(\xv_\beta) 
}{ 
	\prod_{\beta\in\Nalpha} \underline{p}_\beta(\xv_\beta)
} 
\bigg. \notag \\&\quad \bigg.\times 
\frac{ 
	\prod_{\beta\notin\Nalpha} \hat{p}_\beta(\xv_\beta) 
}{ 
	\prod_{\beta\notin\Nalpha} \underline{p}_\beta(\xv_\beta) 
}
 \bigg)
\dif \xv  \\ 
&= \int 
\psi_\alpha(\xv_\alpha) 
\bigg[\prod_{\beta\in\Nalpha} \prod_{\alpha'\in\Nbeta\setminus\alpha} \m{\alpha'}{\beta}(\xv_\beta) \bigg] 
\nonumber\\&\quad\times 
\ln \frac{ 
	\psi_\alpha(\xv_\alpha) 
	\prod_{\beta\in\Nalpha} \prod_{\alpha'\in\Nbeta\setminus\alpha} \m{\alpha'}{\beta}(\xv_\beta)
}{ 
	\prod_{\beta\in\Nalpha} \underline{p}_\beta(\xv_\beta)
}
\dif \xv_\alpha 
\nonumber\\&\quad
+ \sum_{\beta\notin\Nalpha} \int \hat{p}_\beta(\xv_\beta) 
\ln \frac{ \hat{p}_\beta(\xv_\beta) }{ \underline{p}_\beta(\xv_\beta) }
\dif \xv_\beta 
\\ 
&= \sum_{\beta \in \Nalpha} \int 
q_{\alpha,\beta}(\xv_{\beta})
\ln \frac{ q_{\alpha,\beta}(\xv_{\beta}) }{ \underline{p}_{\beta}(\xv_{\beta}) }
\dif \xv_{\beta} 
+ \sum_{\beta\notin\Nalpha} 
D\big(\hat{p}_\beta \big\| \underline{p}_\beta \big)
+ \const 
\\ 
&= \sum_{\beta\in\Nalpha} 
D\big(q_{\alpha,\beta} \big\| \underline{p}_\beta \big)
+ \sum_{\beta\notin\Nalpha} 
D\big(\hat{p}_\beta \big\| \underline{p}_\beta \big)
+ \const 
\label{eq:KLalf2},
\end{align}
where 
\begin{align}
q_{\alpha,\beta}(\xv_{\beta})  \defeq  \int 
\psi_\alpha(\xv_\alpha) 
\bigg[ \displaystyle\prod_{\beta\in\Nalpha} \displaystyle\prod_{\alpha'\in\Nbeta\setminus\alpha}     \m{\alpha'}{\beta}(\xv_\beta) \bigg]  
\dif\xv_{\alpha\setminus\beta} \quad \label{eq:q_alpha_beta}
\end{align}
and $\const$ represents a constant w.r.t. the distribution $\underline{p}$ (which we optimize) whose value is irrelevant and may change at each occurrence.
Equation \eqref{eq:KLalf2} says that, for each $\beta$ in the neighborhood of node $\alpha$, the optimal $\underline{p}_{\beta}$ (i.e., $\hat{p}\new_{\alpha,\beta}$) is \revise{uniquely identified as} the moment match of $q_{\alpha,\beta}$ in the exponential family with sufficient statistics $\phiv_\beta(\xv_\beta)$, where $q_{\alpha,\beta}$ is formed by taking the product of the true factor $\psi_\alpha$ and all the messages impinging on that factor, and then integrating out all variables except $\xv_{\beta}$.
Furthermore, \eqref{eq:malf_new} says that the new message $\m{\alpha}{\beta}\new$ passed from $\alpha$ to $\beta\in\Nalpha$ equals $\hat{p}\new_{\alpha,\beta}$ divided by the product of messages $\{\m{\alpha'}{\beta}\}_{\alpha'\in\Nbeta\setminus\alpha}$, i.e., previous messages to $\beta$ from all directions except $\alpha$. \revise{An illustrative example is shown in Fig.~\ref{fig:factor_graph_EP}. 
\begin{figure}[ht]
	\centering
	\scalebox{.85}{\begin{tikzpicture}
		\tikzstyle{latent}=[circle, minimum size = 7mm, draw =black!80, node distance = 5mm]
		\tikzstyle{factor}=[rectangle, minimum size = 7mm, draw =black!80, node distance = 7mm]
		\tikzstyle{connect}=[draw]
		
		\node[factor] (psi1) {$\psi_a$};
		\node[factor,below=of psi1,text=red, draw=red] (psi2) {$\psi_b$};
		\node[factor,below=of psi2] (psi3) {$\psi_c$};
		
		\node[latent, above right=0mm and 40mm of psi1](x1) {$\rvVec{x}_{1}$};
		\node[latent,below=of x1,text=red, draw=red] (x2) {$\rvVec{x}_2$};	
		\node[latent,below=of x2] (x3) {$\rvVec{x}_3$};	
		\node[latent,below=of x3] (x4) {$\rvVec{x}_4$};		
		
		\path (psi1) [connect] -- node [near start,above=.001pt,sloped ] {${m_{a,1}}$} (x1);
		\path (psi1) [connect] -- node [near start,below,sloped ] {${m_{a,2}}$} (x2);
		\path (psi2) [connect,draw=red] -- node [near start,above=.001pt,sloped, text = red] {${m_{b,2}}$} (x2);
		\path (psi2) [connect] -- node [midway,above=.001pt,sloped ] {$~~{m_{b,3}}$} (x3);
		\path (psi2) [connect] -- node [pos=.45,above=.001pt,sloped ] {${m_{b,4}}$} (x4);
		\path (psi3) [connect] -- node [pos=.18,above=.001pt,sloped ] {${m_{c,3}}$} (x3);
		\path (psi3) [connect] -- node [near start,below,sloped ] {${m_{c,4}}$} (x4);
		\end{tikzpicture}}
	\caption{\revise{An example of the factor graph representation of EP for $\alpha \in \{a,b,c\}$ and $\beta \in \{1,2,3,4\}$. 
	For $\alpha = b$ and $\beta = 2$, according to \eqref{eq:q_alpha_beta}  and \eqref{eq:malf_new}, $q_{b,2}(\rvVec{x}_{2}) = \int \psi_b(\rvVec{x}_2,\rvVec{x}_3,\rvVec{x}_4) m_{a,2}(\rvVec{x}_2) m_{c,3}(\rvVec{x}_3) m_{c,4}(\rvVec{x}_4) \dif \rvVec{x}_3 \dif \rvVec{x}_4$ and $m_{b,2}\new(\rvVec{x}_2) = \frac{\hat{p}\new_{b,2}(\rvVec{x}_2)}{m_{a,2}(\rvVec{x}_2)}$, respectively.}
	}
	\label{fig:factor_graph_EP}
\end{figure}
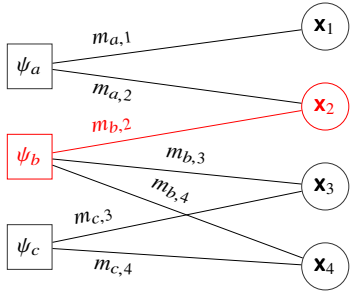
}

\section{Application of EP to Non-Coherent Detection} \label{sec:EP_application}
In order to apply EP to the non-coherent detection problem described in Section~\ref{sec:model}, we express the transmitted signal as $\rvVec{s}_k = \sv_k\of{\rv{i}_k}$, where $\rv{i}_1,\dots,\rv{i}_K$ are independent random indices.\footnote{\revise{The application of EP to non-coherent multi-user detection is non-trivial. Many choices can be made to model and partition the unknowns, but may not result in tractable derivation. Our choice is carefully made to enable closed-form message updates.}} \revise{With the assumption that the constellation symbols are transmitted with equal probability}, $\rv{i}_k$ are uniformly distributed over $[|\Sc_k|]$, $k\in [K]$. We rewrite the received signal~\eqref{eq:channel_model} in vector form as
%
\begin{align}
\rvVec{y} = \sum_{k=1}^K \rvVec{z}_k + \rvVec{w},
\label{eq:y}
\end{align}
where $\rvVec{y}  \defeq {\rm vec}(\rvMat{Y}^\T)$, $\rvVec{z}_k  \defeq (\sv\of{\rv{i}_k}_k \otimes \Id_N) \rvVec{h}_k$, and $\rvVec{w} \defeq {\rm vec}(\rvMat{W}^\T)  \sim\Nc(\mathbf{0},\sigma^2 \Id_{NT})$. The problem of estimating $p_{\rvVec{s}_k|\rvMat{Y}}$ is equivalent to estimating $p_{\rv{i}_k|\rvMat{Y}}$ since they admit the same \revise{PMF}. 

With $\rvVec{z}\defeq{[\rvVec{z}_1^\T,\dots,\rvVec{z}_K^\T]}^\T$ and
$\rvVec{i}\defeq[\rv{i}_1,\dots,\rv{i}_K]^\T$, we can write
\begin{align}
	p_{\rvVec{i},\rvVec{z} | \rvVec{y}}(\iv,\zv|\yv)
	& \propto p_{\rvVec{i},\rvVec{z},\rvVec{y}}(\iv,\zv,\yv) \\
&= p_{\rvVec{y}|\rvVec{z}}(\yv|\zv) p_{\rvVec{z}|\rvVec{i}}(\zv|\iv)p_{\rvVec{i}}(\iv) \\
&= \psi_0(\zv_1,\dots,\zv_K) 
\bigg[\prod_{k=1}^K \psi_{k1}(\zv_{k},{i}_k)\bigg]
\bigg[\prod_{k=1}^K \psi_{k2}(i_k)\bigg]
\label{eq:post},
\end{align}
corresponding to \eqref{eq:p_alf}, where
\begin{align}
\psi_0(\zv_1,\dots,\zv_K) &\defeq p_{\rvVec{y}|\rvVec{z}}(\yv|\zv)
= \Nc\bigg( \yv;\sum_{k=1}^K \zv_k, \sigma^2\Id_{NT} \bigg), \\
\psi_{k1}(\zv_{k},i_k) &\defeq p_{\rvVec{z}_k|\rv{i}_k}(\zv_k)
= \Nc\big(\zv_k; \mathbf{0}, ( \sv\of{i_k}_k\sv\ofH{i_k}_k) \otimes \revise{\Xim_k} \big), \label{eq:psi_k1}\\
\psi_{k2}(i_k) &\defeq p_{\rv{i}_k}(i_k) 
= \frac{1}{|\Sc_k|} \text{~~for~~} i_k\in[|\Sc_k|].
\end{align}

In the following, we consider a realization $\yv$ of $\rvVec{y}$ and use EP to infer the posterior of the indices $\{\rv{i}_k\}$ and, as a by-product, the posterior of $\rvVec{z}_k$, $k\in[K]$.
To do so, we choose the partition $\rvVec{x}=\{\rvVec{z}_k,\rv{i}_k\}_{k=1}^K$ and illustrate the interaction between these variables and the factors $\psi_0,\psi_{k1}$, $\psi_{k2}$ on the bipartite factor graph in Fig.~\ref{fig:factor_graph}. This graph is a tree with a root $\rvVec{y}$ and $K$ leaves $\{\psi_{k2}\}_{k=1}^K$.
\begin{figure}[ht]
	\centering
	\scalebox{.8}{\begin{tikzpicture}
		\tikzstyle{latent}=[circle, minimum size = 7mm, draw =black!80, node distance = 15mm]
		\tikzstyle{obs}=[circle, minimum size = 7mm, fill=gray!40, draw =black!80, node distance = 3mm]
		\tikzstyle{factor}=[rectangle, minimum size = 7mm, draw =black!80, node distance = 15mm]
		\tikzstyle{connect}=[draw]
		
		\node[factor] (psi0) {$ \psi_0 $};			
		\node[obs,left=of psi0] (y) {$ \rvVec{y} $};							
		\path (y) edge [connect] (psi0);
		\node[latent,right=18mm of psi0] (zk) {$\rvVec{z}_k$};			
		\path (psi0) [connect] -- node [midway,above,sloped] {$~~~\underrightarrow{\muv_{k0}, \Cm_{k0}}$} (zk);
		\node[latent,above=10mm of zk] (z1) {$\rvVec{z}_1$};				
		\path (psi0) [connect] -- node [midway,above,sloped] {$\underrightarrow{\muv_{10}, \Cm_{10}}~$} (z1);
		\node[latent,below=10mm of zk] (zK) {$\rvVec{z}_K $};				
		\path (psi0) [connect] -- node [midway,below,sloped] {$\underrightarrow{\muv_{K0}, \Cm_{K0}}$} (zK);
		\node[factor,right=of z1] (psi11) {$\psi_{11}$};	
		\path (z1) [connect] -- node [midway,above ] {$\underleftarrow{\muv_{11}, \Cm_{11}}$} (psi11);
		\node[factor,right=of zk] (psik1) {$\psi_{k1}$};	
		\path (zk) [connect] -- node [midway,above ] {$\underleftarrow{\muv_{k1}, \Cm_{k1}}$} (psik1);
		\node[factor,right=of zK] (psiK1) {$\psi_{K1} $};	
		\path (zK) [connect] -- node [midway,above ] {$\underleftarrow{\muv_{K1}, \Cm_{K1}}$} (psiK1);
		\node[latent,right=of psi11] (i1) {$\rv{i}_1$};			
		\path (psi11) [connect] -- node [midway,above ] {$\underrightarrow{\{\pi_{11}^{(i)}\}_{i=1}^{|\Sc_k|}}$}(i1);
		\node[latent,right=of psik1] (ik) {$\rv{i}_k$};			
		\path (psik1) [connect] -- node [midway,above ] {$\underrightarrow{\{\pi_{k1}^{(i)}\}_{i=1}^{|\Sc_k|}}$}(ik);
		\node[latent,right=of psiK1] (iK) {$ \rv{i}_K $};			
		\path (psiK1) [connect] -- node [midway,above ] {$\underrightarrow{\{\pi_{K1}^{(i)}\}_{i=1}^{|\Sc_k|}}$}(iK);
		\node[factor,right=of i1] (psi12) {$\psi_{12}$};	
		\path (i1) [connect] -- node [midway,above ] {$\underleftarrow{\{\pi_{12}^{(i)}\}_{i=1}^{|\Sc_k|}}$}(psi12);
		\node[factor,right=of ik] (psik2) {$\psi_{k2}$};	
		\path (ik) [connect] -- node [midway,above ] {$\underleftarrow{\{\pi_{k2}^{(i)}\}_{i=1}^{|\Sc_k|}}$}(psik2);
		\node[factor,right=of iK] (psiK2) {$\psi_{K2} $};	
		\path (iK) [connect] -- node [midway,above ] {$\underleftarrow{\{\pi_{K2}^{(i)}\}_{i=1}^{|\Sc_k|}}$}(psiK2);
		
		\node[above=3mm of zk] {$\vdots$};		\node[below=.3mm of zk] {$\vdots$};
		\node[above=3mm of psik1] {$\vdots$};	\node[below=.3mm of psik1] {$\vdots$};
		\node[above=3mm of ik] {$\vdots$};		\node[below=.3mm of ik] {$\vdots$};
		\node[above=3mm of psik2] {$\vdots$};	\node[below=.3mm of psik2] {$\vdots$};
		\end{tikzpicture}}
	\caption{A factor graph representation of the non-coherent detection problem. The messages are depicted with under-arrows showing their direction from a factor node to a variable node.
	}
	\label{fig:factor_graph}
\end{figure}
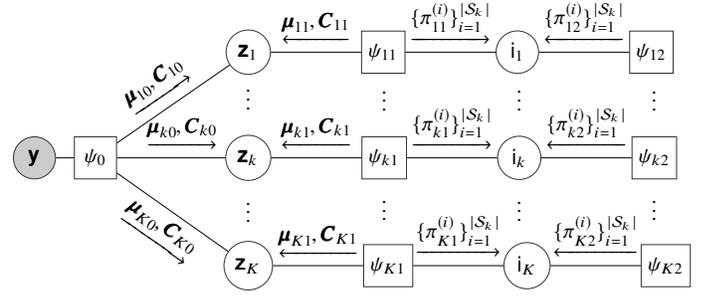

We write the EP approximation according to~\eqref{eq:ptil_beta} as
\begin{align}
\hat{p}_{\rvVec{x}|\rvVec{y}}(\xv|\yv) = \hat{p}_{\rvVec{i},\rvVec{z} | \rvVec{y}}(\iv,\zv|\yv) 
&= \prod_{k=1}^K \hat{p}_{\rvVec{z}_k}(\zv_k)\hat{p}_{\rv{i}_k}(i_k), \label{eq:p_EP_approx1}
\end{align}
where $\hat{p}_{\rvVec{z}_k}(\zv_k)$ and $\hat{p}_{\rv{i}_k}(i_k)$ are implicitly conditioned on $\rvVec{y} = \yv$ and constrained to be a Gaussian vector distribution and a discrete distribution with support $[|\Sc|]$ (both belong to the exponential family), respectively. Specifically, they are parameterized as
\begin{align}
\hat{p}_{\rvVec{z}_k}(\zv_k) 
&= \Nc(\zv_k;\hat{\zv}_k,\Sigmam_k) \text{~~s.t.~~$\Sigmam_k$ is positive definite}, \quad
\label{eq:phat_1k} \\ 
\hat{p}_{\rv{i}_k}(i_k) 
&= \hat{\pi}_{k}\of{i_k} \text{~~for~~} i_k\in[|\Sc_k|]
\text{~~s.t.~~} \sum_{i=1}^{|\Sc_k|} \hat{\pi}_k\of{i}=1. \quad
\label{eq:phat_2k} 
\end{align}
We also write the EP approximation according to~\eqref{eq:ptil_alf} as
\begin{align}
\hat{p}_{\rvVec{x}|\rvVec{y}}(\xv|\yv) &= \hat{p}_{\rvVec{i},\rvVec{z} | \rvVec{y}}(\iv,\zv|\yv) \\
&\propto m_0(\zv_1,\dots,\zv_K) 
\bigg[\prod_{k=1}^K m_{k1} (\zv_k,i_k) \bigg]
\bigg[\prod_{k=1}^K m_{k2} (i_k) \bigg],
\end{align}
where we define
\begin{align}
m_0(\zv_1,\dots,\zv_K) 
&\propto \prod_{k=1}^K \Nc(\zv_k;\muv_{k0},\Cm_{k0}), \\
m_{k1}(\zv_k,i_k) 
&\propto \Nc(\zv_k;\muv_{k1},\Cm_{k1}) \pi_{k1}\of{i_k}, \\
m_{k2}(i_k) 
&= \pi_{k2}\of{i_k} \text{~~for~~} i_k\in[|\Sc_k|]. \\
\end{align}
On the factor graph in Fig.~\ref{fig:factor_graph}, we can interpret 
$(\muv_{k0},\Cm_{k0})$ as the message from factor node $\psi_0$ to variable node $\rvVec{z}_k$, 
$(\muv_{k1},\Cm_{k1})$ as the message from factor node $\psi_{k1}$ to variable node $\rvVec{z}_k$,
$\big\{\pi_{k1}\of{\rv{i}_k}\big\}_{\rv{i}_k=1}^{|\Sc_k|}$ as the message from factor node $\psi_{k1}$ to variable node $\rv{i}_k$, and 
$\big\{\pi_{k2}\of{\rv{i}_k}\big\}_{\rv{i}_k=1}^{|\Sc_k|}$ as the message from factor node $\psi_{k2}$ to variable node $\rv{i}_k$. 
\begin{remark} 
	\revise{Our choice of Gaussian distribution (within the exponential family) in \eqref{eq:phat_1k} is motivated by the fact that when the noise and channel are Gaussian, the symbol posterior takes the form of a Gaussian mixture. It also allows a tractable derivation (using the Gaussian PDF multiplication rule) and closed-form update expressions, as will be shown in the next subsection. If a general (possibly non-Gaussian) channel model is considered, the factor $\psi_{k1}(\zv_k,i_k)$ in \eqref{eq:psi_k1} may be different, but the factor graph in Fig.~\ref{fig:factor_graph} remains unchanged. 
	}
\end{remark}

\subsection{The EP Message Updates}
In the following, we derive the message updates from each of the factor nodes $\psi_0$, $\psi_{k1},$ and $\psi_{k2}$, $k\in [K]$, to the corresponding variable nodes. To do so, for each $\alpha \in \{k1,k2,0\}$, we compute the projected density $\hat{p}\new_{\alpha} = \prod_{k=1}^K \hat{p}\new_{\alpha,\rvVec{z}_k}(\zv_k)\hat{p}\new_{\alpha,\rv{i}_k}({i}_k)$ according to \eqref{eq:p_EP_approx1} and Proposition~\ref{prop:pnew}, and then update the factor $m_\alpha$ according to~\eqref{eq:malf2}. 

\subsubsection{Message $\big\{\pi_{k2}\of{i_k}\big\}_{i_k=1}^{|\Sc_k|}$ from factor node $\psi_{k2}$ to variable node $\rv{i}_k$}
First, we compute $\hat{p}_{k2,\rv{i}_k}\new$ and then the EP message $\big\{\pi_{k2}\of{i_k}\big\}_{i_k=1}^{|\Sc_k|}$ from node $\psi_{k2}$ to node $\rv{i}_k$.
From \eqref{eq:KLalf2} and \eqref{eq:phat_2k}, we know that $\hat{p}_{k2,\rv{i}_k}\new$ is the discrete distribution with \revise{PMF} $\{\hat{\pi}_{k2}\of{i}\}_{i=1}^{|\Sc_k|}$ proportional to $\psi_{k2}({i}_k) \pi_{k1}\of{i_k}$, and so
\begin{align}
\hat{\pi}_{k2}\of{{i}_k}
= \frac{ \psi_{k2}({i}_k) \pi_{k1}\of{{i}_k} }
{ \sum_{i=1}^{|\Sc_k|} \psi_{k2}(i) \pi_{k1}\of{i} } 
= \frac{ \pi_{k1}\of{{i}_k} }
{ \sum_{i=1}^{|\Sc_k|} \pi_{k1}\of{i} } \quad \text{for~${i}_k\in[|\Sc_k|]$},
\label{eq:pihat_k2}  
\end{align}
since $\psi_{k2}({i}_k)$ is constant over these ${i}_k$.
With $\hat{p}\new_{k2,\rv{i}_k}$ computed, \eqref{eq:malf2} implies that the message from node $\psi_{k2}$ to node $\rv{i}_k$ is the \revise{PMF} proportional to
\begin{align}
\frac{\hat{p}\new_{k2,\rv{i}_k}({i}_k)}{\pi_{k1}\of{{i}_k}}
&= \frac{\hat{\pi}_{k2}\of{{i}_k}}{\pi_{k1}\of{{i}_k}}
= \frac{ 1 }{ \sum_{i=1}^{|\Sc_k|} \pi_{k1}\of{i} } 
= \const \quad \text{for ${i}_k\in[|\Sc_k|]$},
\end{align}
and thus 
$
\pi_{k2}\of{{i}_k}
= \frac{1}{|\Sc_k|}
\text{~for~}{i}_k\in[|\Sc_k|].
$

\subsubsection{Messages from factor node $\psi_{k1}$ to variable nodes $\rvVec{z}_k$ and $\rv{i}_k$}
Next, we compute $\hat{p}\new_{k1} = \prod_{k=1}^K \hat{p}\new_{k1,\rvVec{z}_k}(\zv_k)\hat{p}\new_{k1,\rv{i}_k}({i}_k)$ and the messages $\big\{\pi_{k1}\of{{i}_k}\big\}_{{i}_k=1}^{|\Sc_k|}$ and $(\muv_{k1},\Cm_{k1})$ from node $\psi_{k1}$ to nodes $\rv{i}_k$ and $\rvVec{z}_k$, respectively. 

\underline{\it Message $\big\{\pi_{k1}\of{{i}_k}\big\}_{{i}_k=1}^{|\Sc_k|}$ from node $\psi_{k1}$ to node $\rv{i}_k$:}
We first compute $\hat{p}\new_{k1,\rv{i}_k} ({i}_k)$. 
From \eqref{eq:KLalf2} and \eqref{eq:phat_2k}, we know that $\hat{p}\new_{k1,\rv{i}_k}({i}_k)$ is the discrete distribution with support $[|\Sc_k|]$ and \revise{PMF} $\hat{\pi}_{k1}\of{{i}_k}$ proportional to 
\begin{align}
\lefteqn{
	\int \psi_{k1}(\zv_k,{i}_k) \Nc(\zv_k;\muv_{k0},\Cm_{k0}) \pi_{k2}\of{{i}_k} \dif\zv_k
}\nonumber\\
&= \frac{1}{|\Sc_k|}\int
\Nc\big(\zv_k;\mathbf{0},
(\sv_k\of{{i}_k}\sv_k\ofH{{i}_k}) \otimes \revise{\Xim_k}\big) 
\Nc(\zv_k;\muv_{k0},\Cm_{k0}) 
\dif\zv_k\\
&= \frac{1}{|\Sc_k|}\int
\Nc\big(\zv_k;\hat{\zv}_{k{i}_k},\Sigmam_{k{i}_k}\big)
\notag \\ &\quad \times 
\Nc\big(\mathbf{0}; \muv_{k0},
(\sv_k\of{{i}_k}\sv_k\ofH{{i}_k}) 
\otimes \revise{\Xim_k} + \Cm_{k0}\big) \dif\zv_k\\
&= \frac{1}{|\Sc_k|}
\Nc\big(\mathbf{0}; \muv_{k0},
(\sv_k\of{{i}_k}\sv_k\ofH{{i}_k}) 
\otimes \revise{\Xim_k} + \Cm_{k0}\big) ,
\end{align}
where the second equality follows from the Gaussian PDF multiplication rule in Lemma~\ref{lemma:Gmult} 
with
\begin{align}
	\Sigmam_{ki} &= \big( [(\sv_k\of{i}\sv_k\ofH{i}) \otimes \revise{\Xim_k}]^{-1} + \Cm_{k0}^{-1} \big)^{-1} \\
&= \big[(\sv_k\of{i}\sv_k\ofH{i}) \otimes \revise{\Xim_k}\big]\big( (\sv_k\of{i}\sv_k\ofH{i}) \otimes \revise{\Xim_k}
+ \Cm_{k0} \big)^{-1} \Cm_{k0}, \quad
\label{eq:Sigma_ki} \\
	\hat{\zv}_{ki} &= \Sigmam_{ki}\Cm_{k0}^{-1}\muv_{k0} \\
&= \big[(\sv_k\of{i}\sv_k\ofH{i}) \otimes \revise{\Xim_k}\big]\big( (\sv_k\of{i}\sv_k\ofH{i}) \otimes \revise{\Xim_k}
+ \Cm_{k0} \big)^{-1} \muv_{k0}. \qquad
\label{eq:xhat_ki}
\end{align}
Thus
\begin{align}
\hat{\pi}_{k1}\of{{i}_k}
&= \frac{ \Nc\big(\mathbf{0}; \muv_{k0},
	(\sv_k\of{{i}_k}\sv_k\ofH{{i}_k}) 
	\otimes \revise{\Xim_k} + \Cm_{k0}\big) }
{\sum_{i=1}^{|\Sc_k|} \Nc\big(\mathbf{0}; \muv_{k0},
	(\sv_k\of{i}\sv_k\ofH{i}) 
	\otimes \revise{\Xim_k} + \Cm_{k0}\big) }, ~ {i}_k\in[|\Sc_k|].
\label{eq:pihat_k1}
\end{align} 
With $\hat{p}_{k1,\rv{i}_k}\new({i}_k)$ computed, \eqref{eq:malf2} implies that the message $\pi_{k1}\of{{i}_k}$ from node $\psi_{k1}$ to node $\rv{i}_{k}$ is the \revise{PMF} proportional to
$
\frac{\hat{p}\new_{k1,\rv{i}_k}({i}_k)}{\pi_{k2}\of{{i}_k}}
= |\Sc_k|\hat{\pi}_{k1}\of{{i}_k}
\text{~for~}{i}_k\in[|\Sc_k|],
$
and thus 
\begin{align}
\pi_{k1}\of{{i}_k}
&= \frac{|\Sc_k|\hat{\pi}_{k1}\of{{i}_k}}{\sum_{i=1}^{|\Sc_k|} |\Sc_k|\hat{\pi}_{k1}\of{i}}
= \hat{\pi}_{k1}\of{{i}_k}
\text{~~for~~}{i}_k\in[|\Sc_k|] 
\label{eq:p1knew} .
\end{align}

\underline{\it Message $(\muv_{k1},\Cm_{k1})$ from node $\psi_{k1}$ to nodes $\rvVec{z}_k$:}
We next compute $\hat{p}\new_{k1,\rvVec{z}_k}(\zv_k)$. From \eqref{eq:KLalf2} and \eqref{eq:phat_1k}, we know that $\hat{p}\new_{k1,\rvVec{z}_k}(\zv_k)$ is the Gaussian distribution with mean $\hat{\zv}_{k}$ and covariance $\Sigmam_{k}$ matched to that of the PDF proportional to 
\begin{align}
\lefteqn{
	\sum_{{i}_k = 1}^{|\Sc_k|} \psi_{k1}(\zv_k,{i}_k) \Nc(\zv_k;\muv_{k0},\Cm_{k0}) \pi_{k2}\of{{i}_k}
}\nonumber\\
&= \frac{1}{|\Sc_k|}\sum_{i = 1}^{|\Sc_k|}
\Nc\big(\zv_k;\mathbf{0},
(\sv_k\of{i}\sv_k\ofH{i}) \otimes \revise{\Xim_k}\big) 
\Nc(\zv_k;\muv_{k0},\Cm_{k0}) \\
&=  \frac{1}{|\Sc_k|}\sum_{i = 1}^{|\Sc_k|} 
 \Nc\big(\zv_k;\hat{\zv}_{ki},\Sigmam_{ki}\big) 
\Nc\big(\mathbf{0}; \muv_{k0},
(\sv_k\of{i}\sv_k\ofH{i}) 
\otimes \revise{\Xim_k}  +  \Cm_{k0}\big) \\
&\propto 
\sum_{i=1}^{|\Sc_k|} 
\Nc\big(\zv_k;\hat{\zv}_{ki},\Sigmam_{ki}\big)
\hat{\pi}_{k1}\of{i},
\label{eq:GM}
\end{align}
where the second equality follows from the Gaussian PDF multiplication rule in Lemma~\ref{lemma:Gmult} with $\Sigmam_{ki}$ and $\hat{\zv}_{ki}$ defined in \eqref{eq:Sigma_ki} and \eqref{eq:xhat_ki}, respectively. Thus, from \eqref{eq:p1knew}, we have
\begin{align}
\hat{\zv}_k
&= \sum_{i=1}^{|\Sc_k|} {\pi}_{k1}\of{i} \hat{\zv}_{ki},
\label{eq:xhat_k}\\
\Sigmam_k
&= \sum_{i=1}^{|\Sc_k|} {\pi}_{k1}\of{i} (\hat{\zv}_{ki}\hat{\zv}_{ki}^\H + \Sigmam_{ki})
- \hat{\zv}_k\hat{\zv}_k^\H 
\label{eq:Sigma_k} .
\end{align}
With $\hat{p}_{k1,\rvVec{z}_k}\new(\zv_k)$ computed, \eqref{eq:malf2} implies that the message from node $\psi_{k1}$ to node $\rvVec{z}_{k}$ is proportional to
\begin{align}
\frac{\hat{p}\new_{k1,\rvVec{z}_k}(\zv_{k})}
{\Nc(\zv_{k};\muv_{k0},\Cm_{k0})} 
&= \frac{\Nc(\zv_{k};\hat{\zv}_{k},\Sigmam_{k})} 
{\Nc(\zv_{k};\muv_{k0},\Cm_{k0})} 
\propto \Nc(\zv_{k};\muv_{k1},\Cm_{k1}) , \quad
\label{eq:extrinsic}
\end{align}
with
\begin{align}
\Cm_{k1}
&= \big( \Sigmam_{k}^{-1} - \Cm_{k0}^{-1} \big)^{-1},  
\label{eq:C_k1}\\
\muv_{k1}
&= \Cm_{k1} \big( \Sigmam_{k}^{-1}\hat{\zv}_{k} 
- \Cm_{k0}^{-1} \muv_{k0} \big).
\label{eq:mu_k1}
\end{align}
Equations \eqref{eq:C_k1} and \eqref{eq:mu_k1} can be verified using
$\Nc(\zv_{k};\hat{\zv}_{k},\Sigmam_{k})
\propto \Nc(\zv_{k};\muv_{k1},\Cm_{k1})
\Nc(\zv_{k};\muv_{k0},\Cm_{k0})$, which follows from~\eqref{eq:phat_beta}
and the Gaussian PDF multiplication rule in Lemma~\ref{lemma:Gmult}. 

\subsubsection{Message $(\muv_{k0},\Cm_{k0})$ from node $\psi_0$ to node $\rvVec{z}_k$}
Finally, we compute $\hat{p}\new_{0,\rvVec{z}_k}$ and the EP message $(\muv_{k0},\Cm_{k0})$ from node $\psi_0$ to node $\rvVec{z}_k$ for each $k\in[K]$.
From \eqref{eq:KLalf2} and \eqref{eq:phat_1k}, we know that $\hat{p}\new_{0,\rvVec{z}_k}$ is the Gaussian distribution with mean $\hat{\zv}_{k0}$ and covariance $\Sigmam_{k0}$ matched to that of the PDF proportional to 
\begin{align}
\lefteqn{
	\Nc(\zv_k;\muv_{k1},\Cm_{k1})
	\int \psi_0(\zv_{1},\dots,\zv_{K})
	\bigg[ \prod_{j\neq k} \Nc(\zv_{j};\muv_{j1},\Cm_{j1})
	\dif\zv_{j} \bigg]
}\nonumber\\
&= \Nc(\zv_k;\muv_{k1},\Cm_{k1}) 
\notag \\ &\quad \times
\int  
\Nc\bigg(\yv;\zv_{k}
 +  \sum_{j\neq k}\zv_{j},\sigma^2\Id_{NT}\bigg) 
\bigg[ \prod_{j\neq k} \Nc(\zv_{j};\muv_{j1},\Cm_{j1})
\dif\zv_{j} \bigg]  \\
&= \Nc(\zv_{k};\muv_{k1},\Cm_{k1})
\Nc\bigg( \zv_{k}; 
\yv - \sum_{j\neq k}\muv_{j1}, 
\sigma^2\Id_{NT} + \sum_{j\neq k}\Cm_{j1} \bigg) \label{eq:tmp923},
\end{align}
where \eqref{eq:tmp923} follows by applying repeatedly 
 Lemma~\ref{lemma:Gmult}. Applying the Gaussian PDF multiplication rule to \eqref{eq:tmp923}, we obtain
\begin{align}
\Sigmam_{k0}
&= \Big(\Cm_{k1}^{-1} 
+ \Big[\sigma^2\Id_{NT} + \sum_{j\neq k}\Cm_{j1} \Big]^{-1} 
\Big)^{-1}, 
\label{eq:Sigma_k0}\\
\hat{\zv}_{k0}
&= \Sigmam_{k0}
\bigg(
\Cm_{k1}^{-1}\muv_{k1}  + 
\Big[\sigma^2\Id_{NT}  +  \sum_{j\neq k}\Cm_{j1} \Big]^{-1}
\Big[ \yv - \sum_{j\neq k}\muv_{j1} \Big]
\bigg)
\label{eq:xhat_k0} .
\end{align}
Given $\hat{p}\new_{0,\rvVec{z}_k}(\zv_{k})=\Nc(\zv_{k};\hat{\zv}_{k0},\Sigmam_{k0})$, \eqref{eq:malf2} implies that the message from node $\psi_{0}$ to node $\rvVec{z}_{k}$ is proportional to
\begin{align}
\frac{\hat{p}\new_{0,\rvVec{z}_k}(\zv_{k})}
{\Nc(\zv_{k};\muv_{k1},\Cm_{k1})}
&= \frac{\Nc(\zv_{k};\hat{\zv}_{k0},\Sigmam_{k0})}
{\Nc(\zv_{k};\muv_{k1},\Cm_{k1})} 
\propto \Nc(\zv_{k};\muv_{k0},\Cm_{k0}), 
\end{align}
with $\Cm_{k0}  =  \big( \Sigmam_{k0}^{-1} - \Cm_{k1}^{-1} \big)^{-1}$ and $\muv_{k0}
 =  \Cm_{k0} \big( \Sigmam_{k0}^{-1}\zv_{k0} 
- \Cm_{k1}^{-1} \muv_{k1} \big)$. This is verified using
$\Nc(\zv_{k};\hat{\zv}_{k0},\Sigmam_{k0})$
 $\propto  \Nc(\zv_{k};\muv_{k1},\Cm_{k1})
\Nc(\zv_{k};\muv_{k0},\Cm_{k0})$, which follows from~\eqref{eq:phat_beta},
and the Gaussian PDF multiplication rule in Lemma~\ref{lemma:Gmult}.
Plugging in the expressions for $\Sigmam_{k0}^{-1}$ and $\hat{\zv}_{k0}$ from \eqref{eq:Sigma_k0} and \eqref{eq:xhat_k0} yields
\begin{align}
\Cm_{k0}
&= \sigma^2\Id_{NT} + \sum_{j\neq k}\Cm_{j1} ,
\label{eq:C_k0}\\
\muv_{k0}
&= \yv-\sum_{j\neq k}\muv_{j1} 
\label{eq:mu_k0} .
\end{align}
This concludes the derivation of the EP message updates. 

\subsection{Initialization of the EP Messages} \label{sec:EP_initialization}
We initialize the EP messages as follows. First, we choose the non-informative initialization 
$\Cm_{k0}^{-1}=\mathbf{0}$ and $\muv_{k0}=\mathbf{0}$,
so that, from \eqref{eq:pihat_k1},  the initial message from node $\psi_{k1}$ to node $\rv{i}_k$ coincides with the uniform prior
$
{\pi}_{k1}\of{{i}_k} = \hat{\pi}_{k1}\of{{i}_k}
= \frac{1}{|\Sc_k|} \text{~for~} {i}_k\in[|\Sc_k|],
$
and, from \eqref{eq:Sigma_ki} and \eqref{eq:xhat_ki}, the initial parameters $\Sigmam_{ki}
= (\sv_k\of{i}\sv_k\ofH{i}) \otimes \revise{\Xim_k}$ and $\zv_{ki}
= \mathbf{0}$, respectively, for $k\in[K]$ and $i\in[|\Sc_k|]$.
This leads to the initial parameters of $\hat{p}_k(\zv_k)$ from \eqref{eq:xhat_k} and \eqref{eq:Sigma_k} as 
$
\hat{\zv}_{k}
= \mathbf{0} \text{~and~}
\Sigmam_{k}
= \frac{1}{|\Sc_k|}\sum_{i=1}^{|\Sc_k|} (\sv_k\of{i}\sv_k\ofH{i}) \otimes \revise{\Xim_k},
$
and the initial message from node $\psi_{k1}$ to node $\rvVec{z}_{k}$ given in \eqref{eq:C_k1} and \eqref{eq:mu_k1} as
$
\Cm_{k1}
= \Sigmam_k
= \frac{1}{|\Sc_k|}\sum_{i=1}^{|\Sc_k|} (\sv_k\of{i}\sv_k\ofH{i}) \otimes \revise{\Xim_k},  \text{~and~} 
\muv_{k1}
= \hat{\zv}_k
= \mathbf{0}.
$
Finally, the initial messages from node $\psi_0$ to node $\rvVec{z}_k$ follows from \eqref{eq:C_k0} and \eqref{eq:mu_k0} as
$
\Cm_{k0}
= \sigma^2\Id_{NT} + \sum_{j\neq k}
\frac{1}{|\Sc_j|}\sum_{i=1}^{|\Sc_j|} (\sv_j\of{i}\sv_j\ofH{i}) \otimes \revise{\Xim_k}, \text{~and~}
\muv_{k0}
= \yv.
$

\subsection{The Algorithm}
We summarize the proposed EP scheme for probabilistic non-coherent detection in Algorithm~\ref{algo:EP}. 
In the end, according to~\eqref{eq:phat_beta} and \eqref{eq:phat_2k}, the estimated \revise{PMF} 
$\hat{p}_{\rvVec{s}_k|\rvMat{Y}}(\sv_k\of{i_k}|\Ym)$ is given by $\hat{p}_k({i}_k) = \hat{\pi}_k\of{{i}_k} \propto {\pi}_{k1}\of{{i}_{k}}{\pi}_{k2}\of{{i}_k}$, that is $\hat{p}_k({i}_k) = {\pi}_{k1}\of{{i}_{k}}$ since $\pi_{k2}\of{{i}_k}$ is constant. 
 The algorithm goes through the branches of the tree graph in Fig.~\ref{fig:factor_graph} in a round-robin manner. In each branch, the factor nodes are visited 
 from leaf to root. We note that other message passing schedules can be implemented. 

	\IncMargin{.2em}
\begin{algorithm}[h]
	\SetKwData{Left}{left}\SetKwData{This}{this}\SetKwData{Up}{up}
	\SetKwFunction{Union}{Union}\SetKwFunction{FindCompress}{FindCompress}
	\SetKwInOut{Input}{input}\SetKwInOut{Output}{output}

	\SetKwRepeat{Repeat}{repeat}{until}%
	\SetAlgoLined
	\KwIn{the observation $\Ym$; the constellations $\Sc_1,\dots,\Sc_K$;}
	\revise{set the maximal number of iterations $t_{\rm max}$ \;
	initialize of the messages $\{\pi_{k1}\of{i_k}\}_{i_k = 1}^{|\Sc_k|}, \muv_{k1}, \Cm_{k1},\muv_{k0},\Cm_{k0}$, for $k \in [K]$ \;}
	$t \longleftarrow 0$ \;
	\Repeat{\em convergence or $t = t_{\rm max}$}{
		$t \longleftarrow t+1$ \;
		\For{$k\leftarrow 1$ \KwTo $K$}{
			update $\big\{{\pi}_{k1}\of{i_k}\big\}_{i_k = 1}^{|\Sc_k|}$ according to~\eqref{eq:p1knew} and \eqref{eq:pihat_k1} \label{algo:line_update_pmf_ik} \; 
			compute $\{\hat{\zv}_{ki}\}_{i=1}^{|\Sc_k|}$ and $\{\Sigmam_{ki}\}_{i=1}^{|\Sc_k|}$ according to \eqref{eq:xhat_ki} and \eqref{eq:Sigma_ki}, respectively \label{algo:line_update_zhat_Sigma_ki} \; 
			compute $\hat{\zv}_k$ and $\Sigmam_k$ according to \eqref{eq:xhat_k} and \eqref{eq:Sigma_k}, respectively \;
			update $\muv_{k1}$ and $\Cm_{k1}$ according to \eqref{eq:mu_k1} and \eqref{eq:C_k1}, respectively \;
			update $\big\{\muv_{j0}\big\}_{j\ne k}$ and $\big\{\Cm_{j0}\big\}_{j\ne k}$ according to \eqref{eq:mu_k0} and \eqref{eq:C_k0}, respectively \label{algo:line_update_mu0_C0} \; 
		}
	}
	\KwRet{\em The \revise{PMF} $\big\{{\pi}_{k1}\of{i_k}\big\}_{i_k = 1}^{|\Sc_k|}$ of $\hat{p}_{\rvVec{s}_k|\rvMat{Y}}(\sv_k\of{i_k}|\Ym)$ for $k \in [K]$}
	\caption{EP for probabilistic non-coherent detection}
	\label{algo:EP}
\end{algorithm}
\DecMargin{.2em}

In the EP algorithm, \revise{the dominant operation is the update of $\pi_{k1}\of{i_k}$, $\Sigmam_{ki}$, and $\hat{\zv}_{ki}$, which involves the inverse of the $NT\times NT$ matrix  $(\sv_k\of{i_k}\sv_k\ofH{i_k}) \otimes \revise{\Xim_k} + \Cm_{k0}$ (with complexity $O(K^6)$) \emph{for all} $k\in [K]$ and $i_k \in [|\Sc_k|]$.} 
The complexity of computing $\hat{\zv}_k$, $\Sigmam_k$, $\muv_{k1}$, $\Cm_{k1}$, $\big\{\muv_{j0}\big\}_{j\ne k}$, and $\big\{\Cm_{j0}\big\}_{j\ne k}$ are all of lower order. 
Therefore, the complexity per iteration is given by $O(K^72^B)$. In order to reduce this complexity, we derive two simplifications of the EP scheme in the next section.

\section{Simplifications of the EP Detector} \label{sec:MMSE-SIA_EPAK}
\revise{In this section, we attempt to simplify EP by avoiding the inverse of $NT\times NT$ matrices.
}
\subsection{EP with Approximate Kronecker Products~(EPAK)}
\revise{ 
We observe that if $\Cm_{k0}$ could be expressed as a Kronecker product $\bar{\Cm}_{k0} \otimes \revise{\Xim_k}$ with $\bar{\Cm}_{k0} \in \CC^{T\times T}$, we could rewrite ${\pi}_{k1}\of{i_k}$ in \eqref{eq:pihat_k1} as
\begin{align}
{\pi}_{k1}\of{{i}_k}
&= \frac{ \Nc\big(\mathbf{0}; \muv_{k0},
	(\sv_k\of{{i}_k}\sv_k\ofH{{i}_k} + \bar{\Cm}_{k0}) 
	\otimes \revise{\Xim_k}\big) }
{\sum_{i=1}^{|\Sc_k|} \Nc\big(\mathbf{0}; \muv_{k0},
	(\sv_k\of{i}\sv_k\ofH{i} + \bar{\Cm}_{k0}) 
	\otimes \revise{\Xim_k}\big) }. \label{eq:Kron_app_pmf_ik}
\end{align}
Let $\Mm_{k0} \in \CC^{T\times N}$ s.t. $\muv_{k0} = {\rm vec} \big(\Mm_{k0}^\T\big)$, \eqref{eq:Kron_app_pmf_ik} could be computed efficiently using 
\begin{align}
& \Nc\left(\mathbf{0};\muv_{k0},\big(\sv_k\of{i_k}\sv_k\ofH{i_k} +\bar{\Cm}_{k0} \big)\otimes \revise{\Xim_k}\right) \notag\\
&\propto {\big(1 + \sv_k\ofH{i_k}\bar{\Cm}_{k0}^{-1}\sv_k\of{i_k}\big)^{-N}} \notag \\
&\quad \times \exp\bigg(\frac{\revise{\trace\big\{\bar{\Cm}_{k0}^{-1}\sv_k\of{i_k}\sv_k\ofH{i_k} \Mm_{k0} (\Xim_k^{-1})^\T \Mm_{k0}^\H\big\}}}{1+\sv_k\ofH{i_k}\bar{\Cm}_{k0}^{-1}\sv_k\of{i_k}}\bigg)
\label{eq:Kron_app_eff}
\end{align}
since only the $T\times T$ matrix $\bar{\Cm}_{k0}$ needs to be inverted \revise{(the inverse of $\Xim_k$ can be precomputed and stored)}. In general, $\Cm_{k0}$ does not have a Kronecker structure.} Thus we propose to fit $\Cm_{k0}$ to the form of a Kronecker product by solving the least squares problem
\begin{align}
\min_{\bar{\Cm}_{k0} \in \CC^{T\times T}} \|\Cm_{k0} - \bar{\Cm}_{k0} \otimes \revise{\Xim} \|^2_F \label{eq:approxC0}
\end{align}
as formulated in~\cite[Sec.~4]{van1993approximation}. Let $\Cm_{k0}\{i,j\}$ be the $N\times N$ sub-matrix containing the elements in rows from $(i-1)N + 1$ to $iN$ and columns from $(j-1)N+1$ to $jN$ of $\Cm_{k0}$. Let $\bar{c}_{ij}$ be the element in row $i$ and column $j$ of $\bar{\Cm}_{k0}$. It follows that 
\begin{align}
\lefteqn{
	\|\Cm_{k0} - \bar{\Cm}_{k0} \otimes \revise{\Xim_k} \|^2_F
}\nonumber \\
&= \sum_{i=1}^{T} \sum_{j=1}^{T} \|\Cm_{k0}\{i,j\} - \bar{c}_{ij} \revise{\Xim_k} \|^2_F \\
&= \sum_{i=1}^{T} \sum_{j=1}^{T} \|\Cm_{k0}\{i,j\}\|^2_F - \bar{c}_{ij}\trace[\Cm_{k0}\{i,j\}^\H\revise{\Xim_k}] 
\notag \\&\quad
- \bar{c}^*_{ij} \trace[\revise{\Xim_k}\Cm_{k0}\{i,j\}] + |\bar{c}_{ij}|^2\revise{\trace[\Xim_k^2]}.
\end{align}
Observe that $\|\Cm_{k0} - \bar{\Cm}_{k0} \otimes \revise{\Xim_k} \|^2_F$ is the sum of convex quadratic functions of $\bar{c}_{ij}$. Setting the partials $\frac{\partial \|\Cm_{k0} - \bar{\Cm}_{k0} \otimes \revise{\Xim_k} \|^2_F}{\partial \bar{c}_{ij}}$ to zeros, the optimal $\bar{\Cm}_{k0}$ is given by
\begin{align}
\bar{c}_{ij} = \revise{\frac{\trace[\Cm_{k0}\{i,j\} \revise{\Xim_k}]}{\trace[\Xim_k^2]}}.
\end{align}
With the approximation $\Cm_{k0} \approx \bar{\Cm}_{k0} \otimes \revise{\Xim_k}$, we can approximate ${\pi}_{k1}\of{i_k}$ by the RHS of \eqref{eq:Kron_app_pmf_ik}. 
Also, it follows from \eqref{eq:Sigma_ki} and \eqref{eq:xhat_ki} that 
\begin{align}
\Sigmam_{ki} &\approx \big[(\sv_k\of{i_k}\sv_k\ofH{i_k}) \big(\sv_k\of{i_k}\sv_k\ofH{i_k}
+ \bar{\Cm}_{k0} \big)^{-1} \bar{\Cm}_{k0}  \big] \otimes \revise{\Xim_k}, ~~~\label{eq:EPAK:Sigma_ik}\\
\hat{\zv}_{ki} &\approx {\rm vec}\big(\big[\sv_k\of{i_k}\sv_k\ofH{i_k} 
\big(\sv_k\of{i_k}\sv_k\ofH{i_k}
+ \bar{\Cm}_{k0} \big)^{-1} \Mm_{k0}\big]^\T\big). ~~~ \label{eq:EPAK:xhat_ik}
\end{align}
\revise{To compute $\Cm_{k1}$ and $\muv_{k1}$ in \eqref{eq:mu_k1} and \eqref{eq:C_k1}, the inversion of $\Cm_{k0}$ can be simplified as $\Cm_{k0}^{-1} \approx \bar{\Cm}_{k0}^{-1} \otimes \revise{\Xim_k^{-1}}$, but the inverse of $NT\times NT$ matrices involving $\Sigmam_k$ is still required.}

\revise{To keep an accurate message update at early iterations\footnote{\revise{In the uncorrelated fading case, i.e. $\Xim_k = \Id_N$, the approximation of $\Cm_{k0}$ with Kronecker products becomes more accurate when $\hat{\pi}_{k1}$ is closer to a Kronecker-delta distribution, i.e., we have high confidence in one of the symbols. This is likely the case at high SNR after some EP iterations. At early iterations, however, the approximation $\Cm_{k0} \approx \bar{\Cm}_{k0} \otimes \Xim$ can be inaccurate.}},}
let us fix a threshold $t_0 \in  [t_{\max}]$ and modify Algorithm~\ref{algo:EP} as follows. At iteration $t$, if $t\le t_0$, the messages are updated as in lines \ref{algo:line_update_pmf_ik}-\ref{algo:line_update_mu0_C0}; if $t> t_0$, in line~\ref{algo:line_update_pmf_ik}, \eqref{eq:pihat_k1} is replaced by \eqref{eq:Kron_app_pmf_ik} for the update of ${\pi}_{k1}\of{{i}_k}$, and in line~\ref{algo:line_update_zhat_Sigma_ki}, \eqref{eq:xhat_ki} and \eqref{eq:Sigma_ki} are replaced by \eqref{eq:EPAK:xhat_ik} and \eqref{eq:EPAK:Sigma_ik} for the update of $\Sigmam_{ki}$ and $\hat{\zv}_{ki}$, respectively.
We refer to this scheme as EPAK (EP with Approximate Kronecker). It coincides with EP if $t_0 = t_{\max}$.  At iteration $t > t_0$, the dominant operations in EPAK are the inverse of $\sv_k\of{i}\sv_k\ofH{i}
+ \bar{\Cm}_{k0}$ (with complexity $O(K^3)$) in \eqref{eq:EPAK:Sigma_ik} and \eqref{eq:EPAK:xhat_ik} for each $k \in [K]$ and $i \in [|\Sc_k|]$, \revise{and the inverse of $NT\times NT$ matrices (with complexity $O(K^6)$) to compute $\Cm_{k1}$ and $\muv_{k1}$ for each $k \in [K]$}. Thus the complexity at iteration $t$ of EPAK is $O(K^72^B)$ if $t\le t_0$ and $O(K^42^B + K^7)$ if $t > t_{0}$.

\subsection{\revise{Minimum Mean Square Error - Successive Interference Approximation~(MMSE-SIA)}}
\revise{Another method to simplify EP is as follows.} In the EP scheme, as in \eqref{eq:GM} and \eqref{eq:extrinsic}, the message $\Nc(\zv_k;\muv_{k1},\Cm_{k1})$ from node $\psi_{k1}$ to node $\rvVec{z}_k$ is derived by first projecting $\hat{p}\new_{k1,\rvVec{z}_k}(\zv_{k}) \propto
\sum_{i=1}^{|\Sc_k|} {\pi}_{k1}\of{i} \Nc(\zv_k;\hat{\zv}_{ki},\Sigmam_{ki})$ onto the Gaussian family, then dividing the projected Gaussian by $\Nc(\zv_{k};\muv_{k0},\Cm_{k0})$. If we skip the projection of $\hat{p}\new_{k1}(\zv_{k})$ onto the Gaussian family, i.e., we derive $\Nc(\zv_k;\muv_{k1},\Cm_{k1})$ by dividing directly $\hat{p}\new_{k1,\rvVec{z}_k}(\zv_{k})$ to $\Nc(\zv_{k};\muv_{k0},\Cm_{k0})$, 
then the mean $\muv_{k1}$ and covariance matrix $\Cm_{k1}$ are matched to that of the PDF proportional to
\begin{align} 
\frac{\hat{p}\new_{k1,\rvVec{z}_k}(\zv_{k})}{\Nc(\zv_{k};\muv_{k0},\Cm_{k0})} 
&= \sum_{i=1}^{|\Sc_k|} {\pi}_{k1}\of{i} \frac{
	\Nc(\zv_k;\hat{\zv}_{ki},\Sigmam_{ki})}
{\Nc(\zv_k;\muv_{k0},\Cm_{k0})} \\
&\propto \sum_{i=1}^{|\Sc_k|} {\pi}_{k1}\of{i} \Nc\big(\zv_k;\mathbf{0},( \sv_k\of{i}\sv_k\ofH{i})\otimes\revise{\Xim_k}\big) ~~~\label{eq:tmp1051} 
\\ &= \Nc\big(\zv_k;\mathbf{0},
\revise{\Rm_k \otimes\Xim_k} \big).
\end{align}
\revise{where $\Rm_k \defeq \sum_{i=1}^{|\Sc_k|} {\pi}\of{i}_{k1} \sv\of{i}_k\sv_k\ofH{i}$.} \eqref{eq:tmp1051} can be verified using $\Nc(\zv_k;\hat{\zv}_{ki},\Sigmam_{ki}) \propto \Nc\big(\zv_k;\mathbf{0},( \sv_k\of{i}\sv_k\ofH{i})\otimes\revise{\Xim_k}\big)\Nc(\zv_k;\muv_{k0},\Cm_{k0})$, which follows from the Gaussian PDF multiplication rule with $\hat{\zv}_{ki}$ and $\Sigmam_{ki}$ given in \eqref{eq:xhat_ki} and \eqref{eq:Sigma_ki}, respectively. It follows that $\muv_{k1}=\mathbf{0}$ and $\Cm_{k1} = \revise{\Rm_k\otimes\Xim_k}$. As a consequence (see \eqref{eq:mu_k0} and \eqref{eq:C_k0}), $\muv_{k0}=\yv$ and $\Cm_{k0} = \sigma^2\Id_{NT} + \sum_{l\ne k} \revise{\Rm_l\otimes \Xim_k}$. 

This scheme can be alternatively interpreted as follows. 
We expand $\rvVec{y}$ in \eqref{eq:y} as
\begin{align}
\rvVec{y} 
&= (\rvVec{s}_k\otimes \Id_N)\rvVec{h}_k + \sum_{l\ne k} (\rvVec{s}_l\otimes \Id_N)\rvVec{h}_l + \rvVec{w}.
\label{eq:y_vec}
\end{align}
The second term $\rvVec{t}_k \defeq \sum_{l\ne k} (\rvVec{s}_l\otimes \Id_N)\rvVec{h}_l$ is the interference from other users while decoding the signal of user $k$. 
Since the signals $\rvVec{s}_l$ are independent of the channels $\rvVec{h}_l$ and the channels $\rvVec{h}_l$ have zero mean, we have that $\E[\rvVec{t}_k] = \mathbf{0}$.
The covariance matrix of $\rvVec{t}_k$ is
$
\E[\rvVec{t}_k\rvVec{t}_k^\H] 
= \sum_{l\ne k} \E[\rvVec{s}_l \rvVec{s}_l ^\H] \otimes \revise{\Xim_k} 
= \sum_{l\ne k} \Rm_l \otimes \revise{\Xim_k}.
$
If we treat the interference term $\rvVec{t}_k$ as a Gaussian vector with the same mean and covariance matrix\footnote{Another choice is to treat each $\rvVec{s}_l$, $l\ne k$, as a Gaussian. With this choice, however, the interference term $\rvVec{t}_k$ is a product of Gaussians which makes the approximate single-user likelihood difficult to evaluate.}, then $\rvVec{t}_k + \rvVec{w} \sim \Nc\big(\mathbf{0}, \revise{\sum_{l\ne k} \Rm_l\otimes \Xim_k + \sigma^2\Id_{NT}} \big)$. The single-user likelihood under this approximation is
$\hat{p}_{\rvVec{y}| \rvVec{s}_k}(\yv|\sv_k) = \Nc\big(\yv; \mathbf{0},\revise{\sv_k\sv_k^\H \otimes \Xim_k + \sum_{l\ne k} \Rm_l \otimes \Xim_l + \sigma^2\Id_{NT}} \big)$. 
With this and Lemma~\ref{lemma:Gmult}, the update of the approximate posterior $\hat{p}_{\rvVec{s}_k|\rvVec{y}}(\sv_k|\yv) \propto \hat{p}_{\rvVec{y}| \rvVec{s}_k}(\yv|\sv_k)$ coincides with \revise{\eqref{eq:pihat_k1}} for $\muv_{k0}=\yv$ and \revise{$\Cm_{k0} = \sigma^2\Id_{NT} + \sum_{l\ne k} \revise{\Rm_l\otimes \Xim_k}$}. The matrix $\Rm_k$ is then recalculated with the updated value of $\hat{p}_{\rvVec{s}_k|\rvVec{y}}(\sv_k\of{i_k}|\yv)$, $i_k \in [|\Sc_k|]$. The matrices \revise{$\Cm_{l0}$} are updated accordingly, and then used to update $\hat{p}_{\rvVec{s}_l|\rvVec{y}}(\sv_l\of{i_l}|\yv)$, $i_l \in [|\Sc_l|]$, $l \ne k$.

In short, the derived simplification of the EP scheme above iteratively MMSE-estimates the signal $\rvVec{z}_k$ of one user at a time while treating the interference as Gaussian. At each iteration, the Gaussian approximation of the interference for each user is successively improved using the estimates of the signals of other users. We refer to this scheme as MMSE-\revise{SIA} and summarize it in Algorithm~\ref{algo:MMSE-SIA}. In particular, as for the EP scheme, we can start with the non-informative initialization $\hat{p}_{\rvVec{s}_k|\rvMat{Y}}(\sv|\Ym) = \frac{1}{|\Sc_k|}\mathbbm{1}\{\sv \in \Sc_k\}$. 

	\IncMargin{.2em}
	\begin{algorithm}[h] 
		\SetKwData{Left}{left}\SetKwData{This}{this}\SetKwData{Up}{up}
		\SetKwFunction{Union}{Union}\SetKwFunction{FindCompress}{FindCompress}
		\SetKwInOut{Input}{input}\SetKwInOut{Output}{output}

		\SetKwRepeat{Repeat}{repeat}{until}%
		\SetAlgoLined
		\KwIn{the observation $\Ym$; the constellations $\Sc_1,\dots,\Sc_K$;}
		\revise{set the maximal number of iterations $t_{\rm max}$ \;
		initialize of the posteriors $\hat{p}_{\rvVec{s}_k|\rvMat{Y}}(\sv_k|\Ym)$ for $\sv_k \in \Sc_k$, and $\Rm_k = \E_{\hat{p}_{\rvVec{s}_k|\rvMat{Y}}}[\rvVec{s}_k \rvVec{s}_k ^\H]$ for $k \in [K]$ \;}
		$t \longleftarrow 0$ \;
		\Repeat{\em convergence or $t = t_{\rm max}$}{
			$t \longleftarrow t+1$ \;
			\For{$k\leftarrow 1$ \KwTo $K$}{
				compute $\Cm_{k0} = \sigma^2\Id_{NT} + \sum_{l\ne k} \revise{\Rm_l\otimes \Xim_k}$ \label{line:update_Q}\;
				update $\hat{p}_{\rvVec{s}_k|\rvMat{Y}}(\sv_k|\Ym)$, $\sv_k \in \Sc_k$, according to \eqref{eq:pihat_k1} with $\muv_{k0}=\yv$ and \revise{${\Cm}_{k0}$ computed} \label{line:update_phat}\;
				update $\Rm_k = \E_{\hat{p}_{\rvVec{s}_k|\rvMat{Y}}}[\rvVec{s}_k \rvVec{s}_k ^\H]$
				\label{line:update_R} \;
			}
		}
		\KwRet{\em $\hat{p}_{\rvVec{s}_k|\rvMat{Y}}(\sv_k|\Ym)$ for $\sv_k \in \Sc_k$, $k \in [K]$}
		\caption{MMSE-\revise{SIA} for probabilistic non-coherent detection}
		\label{algo:MMSE-SIA}
	\end{algorithm}
	\DecMargin{.2em}

\revise{The complexity order of Algorithm~\ref{algo:MMSE-SIA} is the same as EP due to the $NT\times NT$ matrix inversion in~\eqref{eq:pihat_k1}. However, MMSE-\revise{SIA} still has complexity advantage over EP since no other matrix inversion is required, and there is no need to compute $\{\hat{\zv}_{ki}\}$, $\{\Sigmam_{ki}\}$, $\hat{\zv}_k$, $\Sigmam_k$, or update $\muv_{k1}$. 
If the channel is uncorrelated ($\Xim_k = \Id_N$), the complexity order of MMSE-\revise{SIA} can be reduced. In this case, $\Cm_{k0}$ is the Kronecker product $\Qm_k \otimes \Id_N$ with $\Qm_k \defeq \sum_{l=1,l\ne k}^K\Rm_l + \sigma^2\Id_T$, and thus in \eqref{eq:pihat_k1},  
\begin{align}
& \Nc\big(\mathbf{0}; \muv_{k0},
(\sv_k\of{{i}_k}\sv_k\ofH{{i}_k}) 
\otimes \revise{\Xim_k} + \Cm_{k0}\big) \\ 
&= \Nc\left(\mathbf{0};\rvVec{y},\big(\sv_k\of{i_k}\sv_k\ofH{i_k} +\Qm_k \big)\otimes \Id_N\right) \\
&\propto \big(1+\sv_k\ofH{i_k}\Qm_k^{-1}\sv_k\of{i_k}\big)^{-N} \exp\bigg(\frac{\big\|\rvMat{Y}^\H\Qm_k^{-1}\sv_k\of{i_k}\big\|^2}{1+\sv_k\ofH{i_k}\Qm_k^{-1}\sv_k\of{i_k}}\bigg). ~~~~ \label{eq:tmp1264}
\end{align}
}
Then, only the inverse of $\Qm_k$ is computed, which requires $O(K^3)$ operations. Given $\Qm_k^{-1}$, the complexity of computing the RHS of \eqref{eq:tmp1264} is then $O(K^2)$ for each $i_k \in [|\Sc_k|]$. Therefore, the complexity of computing $\hat{p}_{\rvVec{s}_k|\rvMat{Y}}(\sv_k|\Ym)$ is $O(K^3 + K^22^B)$ for $k \in [K]$. 
Finally, the complexity per iteration of the MMSE-\revise{SIA} algorithm for uncorrelated fading is given by $O(K^4 + K^32^B)$.

\section{Implementation Aspects} \label{sec:implementation}
\subsection{Complexity} \label{sec:complexity}

We summarize the computational complexity of the considered schemes in Table~\ref{tab:complexity}. 

\begin{table}[ht]
	\begin{center}
	\revise{
	\caption{Complexity order of different non-coherent detectors with $T = O(K), N = O(K)$, and $|\Sc_k| = O(2^B)$, $k\in[K]$}
	\centering
	\def\arraystretch{1.2}
	\begin{tabular}{|p{1.9cm}|c|c|}
		\hline
		\multirow{2}{*}{\bf Detector} & \multicolumn{2}{c|}{\bf Complexity order}
		\\ \cline{2-3}
		& Correlated fading & Uncorrelated fading $\Xim_k \!=\! \Id_N, \!\forall k\!$
		\\ \hline \hline
		Optimal (exact marginalization) & $O(K^62^{BK})$ & $O(K^32^{BK})$  \\ \hline
		EP & \multicolumn{2}{c|}{$O(K^72^Bt_{\rm max})$} \\ \hline
		EPAK & \multicolumn{2}{c|}{$O\big(K^72^B t_0 + (K^4 2^B + K^7)(t_{\rm max}-t_0)\big)$} \\ \hline
		MMSE-\revise{SIA} & $O(K^72^Bt_{\rm max})$ & $O(K^4 t_{\rm max} + K^3 2^B t_{\rm max})$ \\ \hline
	\end{tabular}
	\\	\vspace{.1cm}
		\scriptsize $t_{\max}$ denotes the number of iterations. $t_0 \in [t_{\max}]$.
	\label{tab:complexity}
	}
	\end{center}
\end{table}

\subsection{Stabilization}
We discuss some possible numerical problems in the EP algorithm and our solutions. 
\subsubsection{\revise{Singularity of $\Sigmam_k$}}
First, 
in \eqref{eq:Sigma_k}, since the $NT\times NT$ matrix $\Sigmam_k$ is the weighted sum of the terms of rank less than $NT$, it can be close to singular if at a certain iteration, only few of the weights $\pi_{k1}\of{i}$ are sufficiently larger than zero. The singularity of $\Sigmam_k$ can also arise from the constellation structure. For example, the constellations proposed in~\cite{HoangAsilomar2018multipleAccess} are precoded versions of a constellation in $G(\CC^{T-K+1},1)$ and the maximal rank of $\Sigmam_k$ is $N(T-K+1) \le NT$. To avoid the inverse of $\Sigmam_k$, we express $\Cm_{k1}$ in \eqref{eq:C_k1} and $\muv_{k1}$ in \eqref{eq:mu_k1} respectively as
\begin{align}
\Cm_{k1}
&= -\Cm_{k0}\big(\Sigmam_{k}-\Cm_{k0}\big)^{-1}\Sigmam_{k}, \label{eq:tmp1334} \\
\muv_{k1}
&=\Cm_{k0}\big(\Sigmam_{k}  - \Cm_{k0}\big)^{-1} \bigg(\Sigmam_{k}  -  \sum_{i=1}^{|\Sc_k|} \pi_{k1}\of{i} \Sigmam_{ki}\bigg) \Cm_{k0}^{-1} \muv_{k0}.
\end{align}

\subsubsection{\revise{``Negative variance''}} Another problem is that $\Cm_{k1}$
is not guaranteed to be positive definite even if both $\Cm_{k0}$ and $\Sigmam_{k}$ are. When $\Cm_{k1}$ is not positive definite, from \eqref{eq:C_k0}, $\Cm_{k0}$ can have negative eigenvalues, which, through \eqref{eq:pihat_k1}, can make $\hat{\pi}_{k1}\of{{i}_k}$ become close to a Kronecker-delta distribution (even at low SNR) where the position of the mode can be arbitrary, and the algorithm may diverge. Note that this ``negative variance'' problem is common in EP (see, e.g.,~\cite[Sec.~3.2.1]{Minka:Diss:01}, \cite[Sec.~5.3]{vehtari2014expectation}). There has been no generally accepted solution and one normally resorts to various heuristics adapted to each problem. In our problem, to control the eigenvalues of $\Cm_{k1}$, we modify \eqref{eq:tmp1334} by first computing the eigendecomposition $-\Cm_{k0}\big(\Sigmam_{k}-\Cm_{k0}\big)^{-1}\Sigmam_{k} = \Vm \Lambdam \Vm^{-1}$, then computing $\Cm_{k1}$ as 
$
\Cm_{k1} = \Vm |\Lambdam| \Vm^{-1},
$
where $|\Lambdam|$ is the element-wise absolute value of $\Lambdam$. This manipulation of replacing the variance parameters by their absolute values was also used in~\cite{Sun2015iterativeReceiverISI}.

\subsubsection{\revise{Overconfidence at early iterations}}
Finally, due to the nature of the message passing between continuous and discrete distribution, it can happen that all the mass of the \revise{PMF} $\hat{\pi}_{k1}\of{{i}_k}$ is concentrated on a small region of a potentially large constellation $\Sc_k$. For example, if ${\pi}_{k1}\of{{i}_k}$ is close to a Kronecker-delta distribution with a single mode at $i_0$, then \eqref{eq:xhat_ki} and \eqref{eq:Sigma_ki} implies that $\Sigmam_k$ is approximately $\Sigmam_{ki_0}$, and then from \eqref{eq:C_k1}, $\Cm_{k1} \approx (\sv_k\of{i_0}\sv_k\ofH{i_0})\otimes \revise{\Xim_k}$. In this case, almost absolute certainty is placed on the symbol $\sv_k\of{i_0}$, and the algorithm will not be able significantly update its belief in the subsequent iterations. This can be problematic when the mode of ${\pi}_{k1}\of{{i}_k}$ is placed on the wrong symbol at early iterations. To smooth the updates, we apply damping on the update of the parameters of the continuous distributions $\Nc(\zv_k;\muv_{k1},\Cm_{k1})$ and $\Nc(\zv_k;\muv_{k0},\Cm_{k0})$. That is, with a damping factor $\eta \in [0;1]$, at iteration $t$ and for each user $k$, we update 
\begin{align}
\Cm_{k1}(t)
&
= \eta \Vm(t) |\Lambdam(t)| \Vm^{-1}(t) + (1-\eta)\Cm_{k1}(t-1), \label{eq:C_k1_damped} 
\\
\muv_{k1}(t)
&
= \eta \Cm_{k0}(t - 1)\big(\Sigmam_{k}(t) - \Cm_{k0}(t - 1)\big)^{-1} 
\notag \\&\quad \times 
\bigg(\Sigmam_{k}(t)  -  \sum_{i=1}^{|\Sc_k|} \pi_{k1}\of{i}(t) \Sigmam_{ki}(t)\bigg) \Cm_{k0}^{-1}(t \!-\! 1) \muv_{k0}(t \!-\! 1) \notag \\
&\quad+ (1-\eta) \muv_{k1}(t-1), \label{eq:mu_k1_damped} \\
\Cm_{l0}(t)
&= \eta\bigg(\sigma^2\Id_{NT} + \sum_{j\neq l}\Cm_{j1}(t) \bigg) + (1-\eta)\Cm_{l0}(t-1), 
\notag \\&
\quad \forall l \ne k, \label{eq:C_k0_damped}
\\
\muv_{l0}(t)
&= \eta\bigg(\yv- \sum_{j\neq l}\muv_{j1}(t)\bigg) + (1-\eta)\muv_{l0}(t - 1), \quad \forall l \ne k. \label{eq:mu_k0_damped}
\end{align}

In short, we stabilize the EP message updates by replacing \eqref{eq:C_k1_damped}, \eqref{eq:mu_k1_damped}, \eqref{eq:C_k0_damped}, and \eqref{eq:mu_k0_damped} for \eqref{eq:C_k1}, \eqref{eq:mu_k1}, \eqref{eq:C_k0}, and \eqref{eq:mu_k0}, respectively. This technique also applies to EPAK. For MMSE-\revise{SIA}, we damp the update of $\Qm_k$ and $\Rm_k$ in a similar manner as $\Qm_k(t) = \eta\big(\sum_{l\ne k}\Rm_l(t-1) + \sigma^2\Id_T\big) + (1-\eta)\Qm_k(t-1)$ and $\Rm_k(t) = \eta \sum_{i_k =1}^{|\Sc_k|}\pi_{k1}\of{i_k}(t)\sv_k\of{i_k} \sv_k\ofH{i_k} + (1-\eta)\Rm_k(t-1).$ Note that damping does not change the complexity order of these schemes.
The approaches described in this subsection were implemented for the numerical results in the next section.

\section{Performance Evaluation} \label{sec:performance}
In this section, we evaluate the performance of our proposed schemes for a given set of individual constellations. We assume that $B_1 = \dots B_K \eqdef B$.
\revise{We consider the local scattering model~\cite[Sec. 2.6]{Emil2017_massivemimobook} for the correlation matrices $\Xim_k$. Specifically, the $(l,m)$-th element of $\Xim_k$ is generated as
$
[\Xim_k]_{l,m} = \xi_k \E_{\delta_k}[\exp(2\pi d_H (l-m)\sin(\varphi_k + \delta_k))],
$
where $d_H$ is the antenna spacing in the receiver array (measured in number of wavelengths), $\varphi_k$ is a deterministic nominal angle, and $\delta_k$ is a random deviation. We consider $d_H = \frac12$, $\varphi_k$ generated uniformly in $[-\pi,\pi]$, and $\delta_k$ uniformly distributed in $[-\sqrt{3}\sigma_\varphi, \sqrt{3}\sigma_\varphi]$ with angular standard deviation $\sigma_\varphi = 10^\circ$. We also consider $\xi_k = 1, \forall k.$ We set a damping factor $\eta = 0.9$ for EP, EPAK, and MMSE-SIA.}

\subsection{Test Constellations, State-of-the-Art Detectors, and Benchmarks}
\subsubsection{\revise{Precoding-Based Grassmannian Constellations}}
We consider the constellation design in~\cite{HoangAsilomar2018multipleAccess}, which imposes a geometric separation between the individual constellations through a set of precoders $\Um_k$, $k\in[K]$. Specifically, starting with a Grassmannian constellation 
$\Dc = \big\{\dv\of{1},\dots,\dv\of{2^B}\}$ in $G(\CC^{T-K+1},1)$, the individual constellation $\Sc_k$ is generated as 
\begin{align}
\sv_k\of{i} = \frac{\Um_k\dv\of{i}}{\|\Um_k\dv\of{i}\|}, \quad i\in [2^B]. 
\end{align}
We consider the precoders $\Um_k$ defined in~\cite[Eq.(11)]{HoangAsilomar2018multipleAccess} and two candidates for $\Dc$:
\begin{itemize}[leftmargin=*]
	\item A numerically optimized constellation generated by solving the max-min distance criteria
	\begin{align} \label{eq:max_min_chordal}
	\max_{\dv\of{i} \in G(\CC^{T-K+1},1), i = 1,\dots, 2^B} \ \min_{1\le i < j \le 2^B} d(\dv\of{i},\dv\of{j}), \quad
	\end{align}
	where $d(\dv\of{i},\dv\of{j}) \defeq \sqrt{1-|\dv\ofH{i}\dv\of{j}|^2}$ is the chordal distance between two Grassmannian points represented by $\dv\of{i}$ and $\dv\of{j}$. A constellation with maximal minimum pairwise distance leads to low symbol error rate in the absence of the interference. In our simulation, we approximate \eqref{eq:max_min_chordal} by $\min_{\Dc} \log \sum_{1\le i < j \le 2^B} \exp \big(\frac{|\dv\ofH{i}\dv\of{j}|}{\epsilon}\big)$ with a small $\epsilon$ for smoothness, then solve it using gradient descent on the Grassmann manifold using the Manopt toolbox~\cite{manopt}. 
	
	\item The cube-split constellation proposed in~\cite{hoangAsilomar2017cubesplit,Hoang_cubesplit_journal}. This structured constellation has good distance properties and allows for low-complexity single-user decoding and a simple yet effective binary labeling scheme.
\end{itemize}

Exploiting the precoder structure, \cite{HoangAsilomar2018multipleAccess} introduced a  detector~\cite[Sec.~V-B-3]{HoangAsilomar2018multipleAccess} that iteratively mitigates interference by projecting the received signal onto the subspace orthogonal to the interference subspace. We refer to it as POCIS~(Projection onto the Orthogonal Complement of the Interference Subspace). For each user $k$, POCIS first estimates the row space of the interference $\sum_{l\ne k} \rvVec{s}_l \rvVec{h}_l^\T$ based on the precoders and projects the received signal onto the orthogonal complement of this space. It then performs single-user detections to obtain point estimates of the transmitted symbols. From these estimates, POCIS estimates the column space of the interference and projects the received signal onto its orthogonal complement. This process is repeated in the next iteration. 
The complexity order of POCIS is equivalent to the MMSE-\revise{SIA} scheme.
Note that only the indices of the estimated symbols are passed in POCIS, as opposed to the soft information on the symbols as in EP, MMSE-\revise{SIA}, and EPAK. 

\revise{
\subsubsection{Pilot-Based Constellations} \label{sec:pilot-based}
We also consider the pilot-based constellations in which the symbols are generated as $\sv_k\of{i} = \Big[\sqrt{\frac{K}{T}}\ev_k^\T \ \sqrt{\frac{T-K}{T P_{\rm avg}}}\tilde{\sv}_k\ofT{i}\Big]^\T$ where $\ev_k$ is the $k$-th column of $\Id_K$, $\tilde{\sv}_k\of{i}$ is a vector of data symbols taken from a scalar constellation, such as QAM, and $P_{\rm avg}$ is the average symbol power of the considered scalar constellation. Note that this corresponds to the scenario where the $K$ users transmit mutually orthogonal pilot sequences, followed by spatially multiplexed parallel data transmission. Many MIMO detectors have been proposed specifically for these constellations. We consider some representatives as follows.
\begin{itemize}[leftmargin=*]
	\item 
	The receiver MMSE-estimates the channel based on the first $K$ rows of $\rvMat{Y}$, then MMSE-equalizes the received data symbols in the remaining $T-K$ rows of $\rvMat{Y}$, and performs a scalar demapper on the equalized symbols.
	
	\item The receiver MMSE-estimates the channel, then decodes the data symbols using the  Schnorr-Euchner sphere decoder~\cite{Schnorr1994lattice}, referred to as SESD.
	
	\item The receiver performs the semi-blind joint ML channel estimation and data detection scheme in \cite{Hanzo2008semiblind} with repeated weighted boosting search~(RWBS) for channel estimation and the Schnorr-Euchner sphere decoder for data detection, referred to as RWBS-SESD.
\end{itemize}
We note that the sphere decoder has near optimal performance given the channel knowledge, but its complexity is non-deterministic and can be exponential in the channel dimension if the channel matrix is ill-conditioned. 
}

\subsubsection{\revise{Benchmarks}}
We consider the optimal ML detector, whenever it is feasible, as a benchmark. When the optimal detector is computationally infeasible, we resort to another benchmark consisting in giving the receiver, while it decodes the signal $\rvVec{s}_k$ of user $k$, the knowledge of the signals $\rvVec{s}_l$ (but not the channel $\rvVec{h}_l$) of all the interfering users $l \ne k$. With this genie-aided information, optimal ML decoding~\eqref{eq:MLdecoder} can be performed by keeping $\rvVec{s}_l$ fixed for all $l\ne k$ and searching for the best $\rvVec{s}_k$ in $\Sc_k$, thus reducing the total search space size from $2^{BK}$ to $K2^B$. The posterior marginals are computed separately for each user accordingly. This genie-aided detector gives an upper bound on the performance of EP, MMSE-\revise{SIA}, EPAK, and POCIS. 



\revise{
\subsection{Convergence and Running Time}
To assess the convergence of the algorithms, we evaluate the total variation distance between the estimated marginal posteriors $\hat{p}_{\rvVec{s}_k|\rvMat{Y}}$ at each iteration and the exact marginal posteriors $p_{\rvVec{s}_k|\rvMat{Y}}$ when exact marginalization \eqref{eq:exact_marginalization} is possible. The total variation distance between two probability measures $P$ and $Q$ on $\Xc$ is defined as $\Tc\Vc(P,Q) \defeq \frac{1}{2}\sum_{x\in\Xc} |P(x) - Q(x)|$. At iteration $t$ where the estimated posteriors are $\hat{p}^{(t)}_{\rvVec{s}_k|\rvMat{Y}}$, $k\in [K]$, we evaluate the average total variation distance as
\begin{align}
\Delta_t = \frac{1}{K} \sum_{k=1}^{K} \E_{\rvMat{Y}}[\Tc\Vc(\hat{p}^{(t)}_{\rvVec{s}_k|\rvMat{Y}},p_{\rvVec{s}_k|\rvMat{Y}})].
\end{align}

We consider the precoding-based Grassmannian constellations. Fig.~\ref{fig:totalVar} shows the empirical average total variation $\Delta_t$ for $T=6$, $K = 3$, $N = 4$, and $B = 4$ at $\SNR = 8$~dB. As can be seen, at convergence, EP provides the most accurate estimates of the marginal posteriors although it is less stable than other schemes. EP converges after $6$ iterations while MMSE-SIA converges after $5$ iterations. For uncorrelated fading, EPAK with $t_0=2$ can be eventually better than MMSE-SIA, but converges slower. For correlated fading, EPAK totally fails because of the inaccuracy of the approximation with Kronecker products. POCIS converges very quickly after $2$ iterations but achieves a relatively low accuracy of the posterior estimation.  
\begin{figure}[!th] 
	\centering
	\hspace{-.3cm}
	\subfigure[Uncorrelated fading $\Xim_k = \Id_N$]{\includegraphics[width=.52\textwidth]{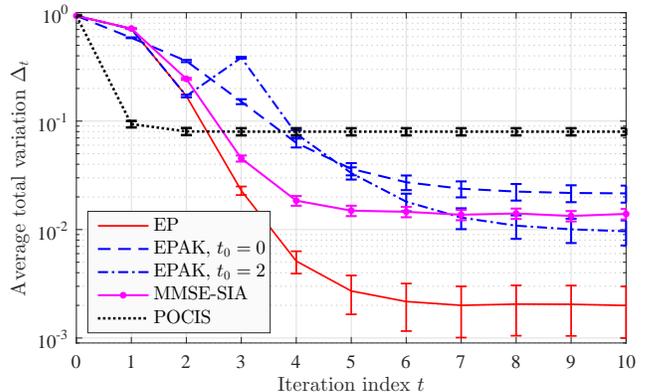}}
	\hspace{-.7cm}
	\subfigure[Correlated Fading]{\includegraphics[width=.52\textwidth]{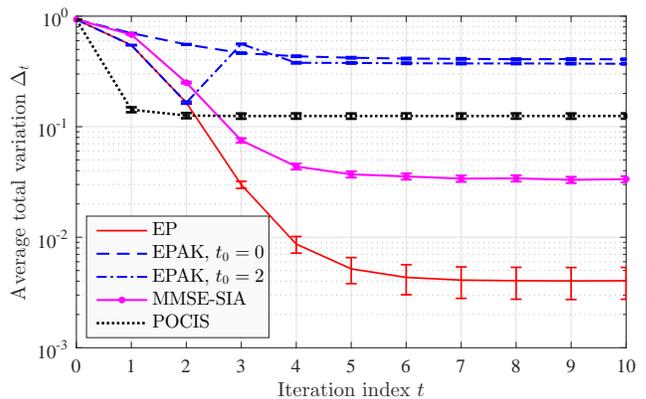}}
	\hspace{-.5cm}
	\caption{The empirical average total variation $\Delta_t$ over $1000$ realizations of the transmitted signal, channel, and noise versus iteration for different non-coherent soft detection schemes for $T=6$, $K = 3$, $B = 4$, and $N = 4$ at $\SNR = 8$~dB. The error bars show the standard error, which is the standard deviation normalized by the square root of the number of samples. For correlated fading, these figures are further averaged over $10$ realizations of the correlation matrices.
	}
	\label{fig:totalVar}
\end{figure}

Fig.~\ref{fig:running_time} depicts the average running time (on a local server) of exact marginalization compared with $6$ iterations of EP, EPAK, MMSE-SIA, and POCIS at $\SNR = 8$~dB. These schemes have significantly lower computation time than exact marginalization. The running time saving of EPAK w.r.t. EP is not significant, even with $t_0 = 0$.
For uncorrelated fading, MMSE-SIA has much shorter running time than all other schemes. 
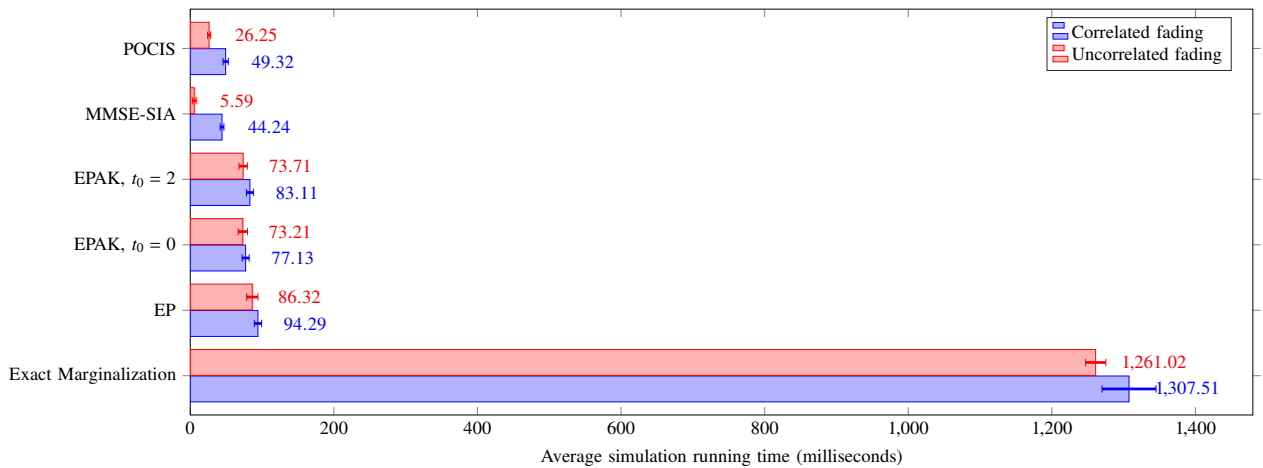
\begin{figure*}[th!] 
	\centering
	\begin{tikzpicture}[scale=.7]
	\begin{axis}[
	xbar = .3,
	xmin=0,
	xmax=1480,
	width=1.2\textwidth,
	height=0.38\textheight,
	enlarge y limits = 0.12,
	bar width=14pt,
	xlabel={Average simulation running time (milliseconds)},
	symbolic y coords={{Exact Marginalization},{EP},{EPAK, $t_0=0$},{EPAK, $t_0=2$},{MMSE-SIA},{POCIS}},
	xtick={0,200,400,600,800,1000,1200,1400},
	nodes near coords,
	every node near coord/.append style={xshift=0.37cm},
	legend style={ anchor=north east, legend cell align=left, align=left, draw=white!15!black, 
	}
	]
	
	\addplot[fill=blue!30, draw=blue, nodes near coords style={color=blue}, error bars/.cd,
	x dir=both,
	x explicit,
	error bar style={line width=1.5pt,color=blue}] coordinates {
		(1307.509645800003,Exact Marginalization) +- (0037.611235654106,0)
		(0094.2926294,EP) +- (0005.113151754871,0)
		(0077.1268068,{EPAK, $t_0=0$}) +- (0004.943069922073,0)
		(0083.1117580,{EPAK, $t_0=2$}) +- (0004.840391362472,0)
		(0044.24469840,MMSE-SIA) +- (0002.635243995103,0)
		(0049.31881420,POCIS) +- (0003.671752066412,0)
	};
	\addlegendentry{Correlated fading}
	
	\addplot[fill=red!30, draw=red, nodes near coords style={color=red}, error bars/.cd,
	x dir=both,
	x explicit,
	error bar style={line width=1.5pt,color=red}] coordinates {
		(1261.015342000001,Exact Marginalization) +- (0014.167520097206,0)
		(0086.315282000000,EP) +- (0007.888196338832,0)
		(0073.2118220,{EPAK, $t_0=0$}) +- (0006.491571935403,0)
		(0073.711169000,{EPAK, $t_0=2$}) +- (0005.849913196009,0)
		(0005.59252900,MMSE-SIA) +- (0002.669959467617,0)
		(0026.24637900,POCIS) +- (0002.140511977281,0)
	};
	\addlegendentry{Uncorrelated fading}
	
	\end{axis}
	\end{tikzpicture}
	\caption{The average running time over $1000$ realizations of the transmitted signal, channel, and noise of exact marginalization vs. $6$ iterations of the considered detection schemes for $T=6$, $K = 3$, $B = 4$, and $N = 4$ at $\SNR = 8$~dB. The error bars show the standard deviation. For correlated fading, the running time is further averaged over $10$ realizations of the correlation matrices.}
	\label{fig:running_time}
		\vspace{-.cm}
\end{figure*}

From these convergence behaviors, hereafter, we fix the number of iterations of EP, MMSE-SIA, and EPAK as $6$ and of POCIS as $3$. Furthermore, we consider EPAK only for uncorrelated fading. For correlated fading, we generate the correlation matrices once and fix them over the simulation.
}

\subsection{Achievable Rate}
We first plot the achievable mismatched sum-rate $\GMI$ of the system calculated as in~\eqref{eq:approx_achievableRate} for $T = 6$, $K = 3$, $N =4$ and $B \in \{4,8\}$ in Fig.~\ref{fig:Rate}. \revise{We consider the precoding-based Grassmannian constellations. For $\Dc$, we use the numerically optimized constellation if $B= 4$ and the cube-split constellation if $B=8$. For uncorrelated fading (Fig.~\ref{fig:Rate_uncorrelated}, the rates achieved with EP and MMSE-SIA detectors are very close to the achievable rate of the system (with the optimal detector) and not far from that of the genie-aided detector. EPAK (with $t_0 = 2$) achieves a very low rate, especially in the low SNR regime where the Kronecker approximation is not accurate. For correlated fading, (Fig.~\ref{fig:Rate_correlated}), the rates achieved with EP and MMSE-SIA are only marginally lower than that of the optimal detector and genie-aided detector. In both cases, the rate achieved with POCIS is lower than that of EP and MMSE-SIA in the lower SNR regime and converges slowly with SNR to the limit $\frac{BK}{T}$ bits/channel use.} 

\begin{figure}[!h]  
	\centering
	\hspace{-.5cm}
	\subfigure[Uncorrelated fading]{\includegraphics[width=.52\textwidth]{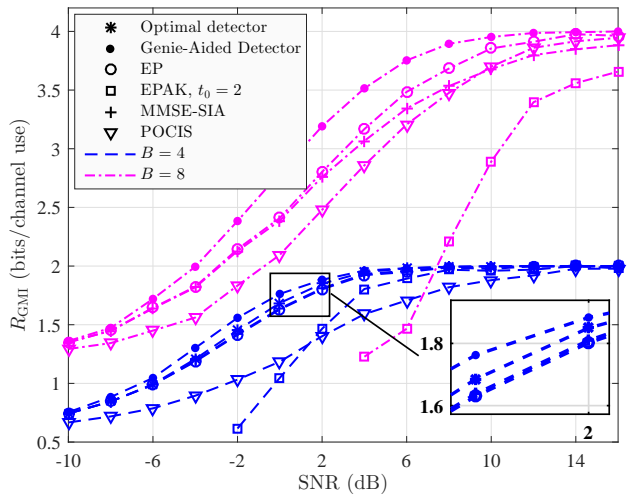}
	\label{fig:Rate_uncorrelated}}
	\hspace{-.9cm}
	\subfigure[Correlated fading]{\includegraphics[width=.52\textwidth]{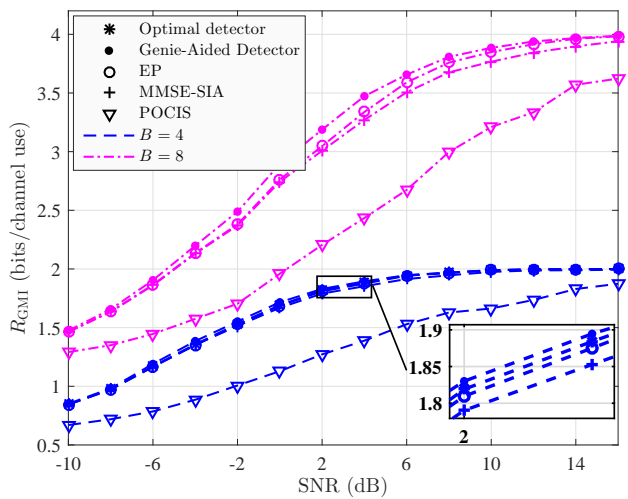} \label{fig:Rate_correlated}}
	\hspace{-.7cm}
	\caption{The mismatched rate of the system with EP, EPAK (with $t_0 = 2$), MMSE-\revise{SIA}, and POCIS detectors in comparison with the optimal detector and/or the genie-aided detector for $T=6$, $K=3$, $N = 4$, and $B \in \{4,8\}$. 
	}
	\label{fig:Rate}
\end{figure}

\subsection{Symbol Error Rates of Hard Detection}
Next, we use the outputs of EP, EPAK, MMSE-\revise{SIA} and POCIS for a maximum-a-posteriori~(MAP) hard detection. We evaluate the performance in terms of symbol error rate~(SER).

In Fig.~\ref{fig:SER_T6K3_B4}, we consider \revise{the precoding-based constellations} with $T = 6$, $K = 3$, $N \in \{4,8\}$, and $B =4$, for which the optimal ML detector~\eqref{eq:MLdecoder} is computationally feasible. We observe that the SER of the EP and MMSE-\revise{SIA} detectors are not much higher than that of the optimal detector, especially in the lower SNR regime. The SER of EPAK is significantly higher than that of EP and MMSE-\revise{SIA} for $t_0 = 0$. This is greatly improved by setting $t_0=2$, i.e., keeping the first two iterations of EP. The gain of EP w.r.t. EPAK and MMSE-\revise{SIA} is more pronounced when the SNR increases. 
\revise{For correlated fading, EP performs almost as good as the optimal detector, whose SER performance is closely approximated by the genie-aided detector.}
\begin{figure}[!h] 
	\centering
	\hspace{-.5cm}
	\subfigure[Uncorrelated fading]{\includegraphics[width=.52\textwidth]{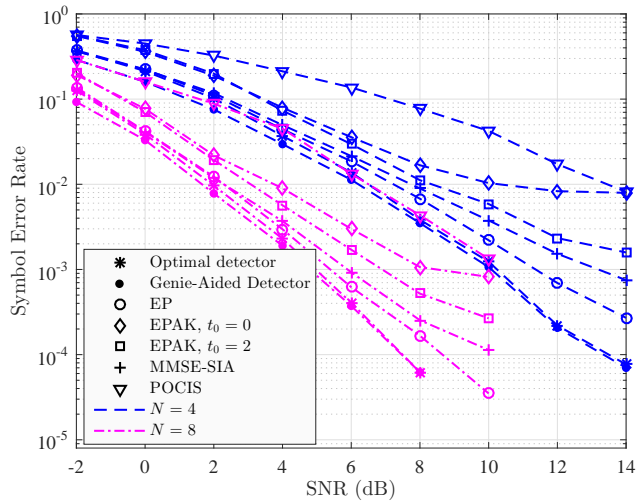}}
	\hspace{-.7cm}
	\subfigure[Correlated fading]{\includegraphics[width=.52\textwidth]{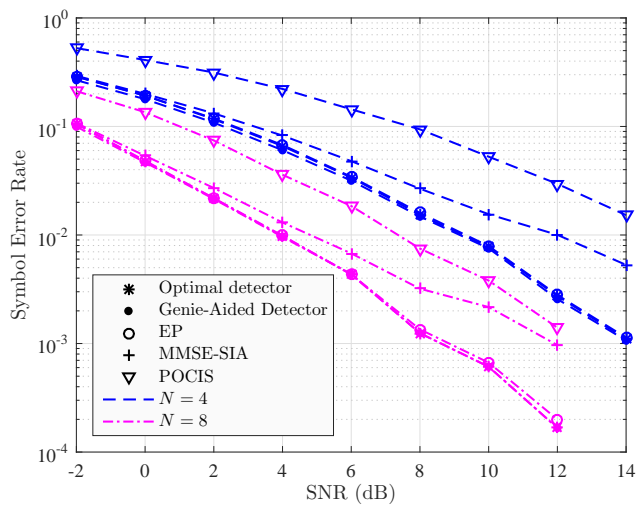}}
	\hspace{-.7cm}
	\caption{The symbol error rate of the system with EP, EPAK (with $t_0 \in \{0,2\}$), MMSE-\revise{SIA}, and POCIS detectors in comparison with the optimal detector and the genie-aided detector for $T=6$, $K=3$, $N \in \{4,8\}$ and $B = 4$. 
	}
	\label{fig:SER_T6K3_B4}
\end{figure}

In Fig.~\ref{fig:SER_T6K3_B9N8}, we consider $T = 6$, $K = 3$, $N = 8$, and $B =9$ and use the genie-aided detector as a benchmark. \revise{In Fig.~\ref{fig:SER_T6K3_B9N8_uncorrelated}, we consider uncorrelated fading and use the pilot-based constellations with $8$-QAM data symbols.} The performance of EP is very close to that of the genie-aided detector. The performance of MMSE-\revise{SIA} is close to EP in the low SNR regime ($\SNR \le 8$~dB). 
We also depict the SER of \revise{the three pilot-based detectors in Section~\ref{sec:pilot-based}, namely, 1) MMSE channel estimation, MMSE equalizer, and QAM demapper, 2) SESD, and 3) RWBS-SESD.
These three schemes are outperformed by the EP detectors. In Fig.~\ref{fig:SER_T6K3_B9N8_correlated}, we consider correlated fading and use the precoding-based Grassmannian constellations with $\Dc$ numerically optimized. We observe again that EP achieves almost the same SER performance as the genie-aided detector.} 
\begin{figure}[!h] 
	\centering
	\vspace{-.4cm}
	\subfigure[Uncorrelated fading, pilot-based constellations]{\includegraphics[width=.52\textwidth]{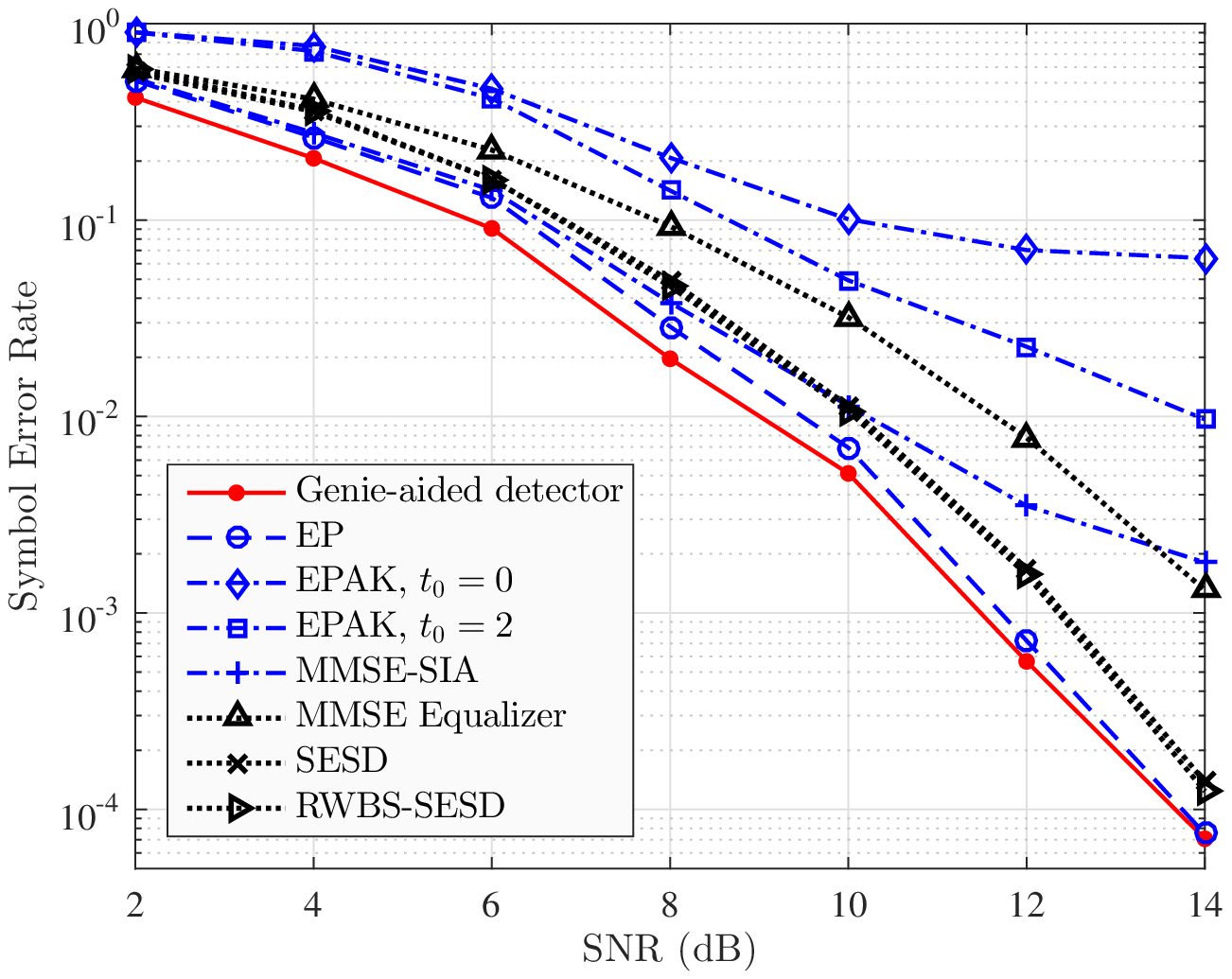} \label{fig:SER_T6K3_B9N8_uncorrelated}}
	\subfigure[Correlated fading, precoding-based constellations]{\includegraphics[width=.52\textwidth]{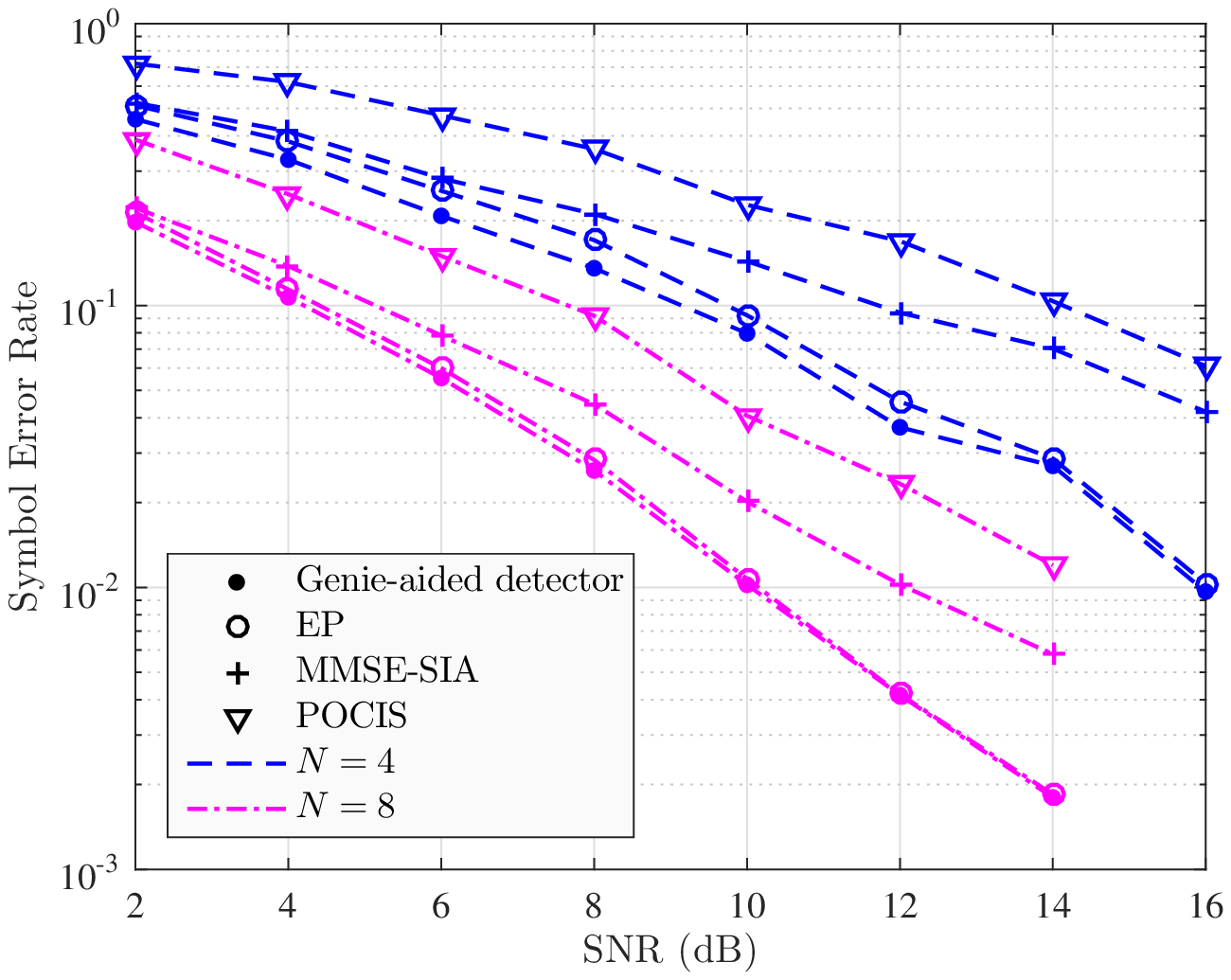} \label{fig:SER_T6K3_B9N8_correlated}}
	\caption{The symbol error rate of the system with EP, EPAK  ($t_0 \in \{0,2\}$), MMSE-\revise{SIA}, POCIS vs. the genie-aided detector for $T=6$, $K=3$, $N = 8$, and $B = 9$. For uncorrelated fading, these schemes are compared with three pilot-based detectors using respectively MMSE equalizer, sphere decoding~\cite{Schnorr1994lattice}, and joint channel estimation--data detection~\cite{Hanzo2008semiblind}.}
	\label{fig:SER_T6K3_B9N8}
\end{figure}

\subsection{Bit Error Rates with a Channel Code}
In this subsection, \revise{we use the output of the soft detectors for channel decoding. We consider the precoding-based Grassmannian constellations with the cube-split constellation for $\Dc$ since it admits an effective and simple binary labeling~\cite{Hoang_cubesplit_journal}.} We take the binary labels of the symbols in $\Dc$ for the corresponding symbols in $\Sc_k$. We integrate a standard symmetric parallel concatenated rate-$1/3$ turbo code~\cite{3GPP_TS36_212}. The turbo encoder accepts packets of $1008$
bits; the turbo decoder computes the bit-wise LLR from the soft outputs of the detection scheme \revise{as in \eqref{eq:LLR}} and performs 10 decoding iterations for each packet.

In Fig.~\ref{fig:BER_turboK}, we show the bit error rate~(BER) with this turbo code using $B = 8$ bits/symbol and different values of $T$ and $K=N$. EP achieves the closest performance to the genie-aided detector and the optimal detector~\eqref{eq:exact_marginalization}. The BER of MMSE-\revise{SIA} vanishes slower with the SNR than the other schemes, and becomes better than POCIS as $K$ and $N$ increase. The BER of EPAK with $t_0 = 2$ is higher than all other schemes. \revise{Under uncorrelated fading}, for $T=7$ and $K =N=4$, the power gain of EP w.r.t. MMSE-\revise{SIA}, POCIS, and EPAK for the same BER of $10^{-3}$ is about $3$~dB, $4$~dB, and $8$~dB, respectively. We also observe that the genie-aided detector gives very optimistic BER performance results compared to the optimal detector.
\begin{figure}[!h] 
	\centering
	\hspace{-.5cm}
	\subfigure[Uncorrelated fading]{\includegraphics[width=.52\textwidth]{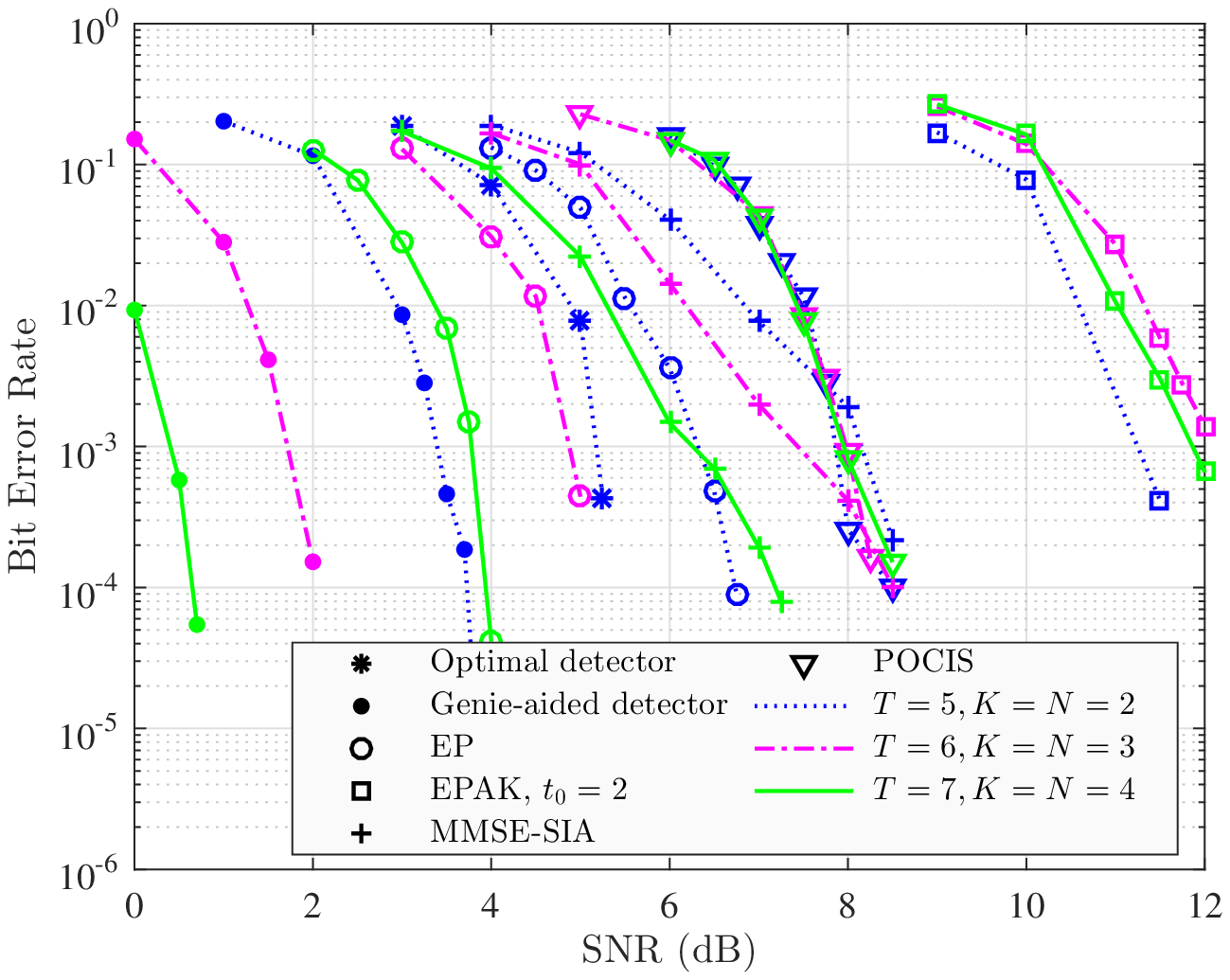}}
	\hspace{-.9cm}
	\subfigure[Correlated fading]{\includegraphics[width=.52\textwidth]{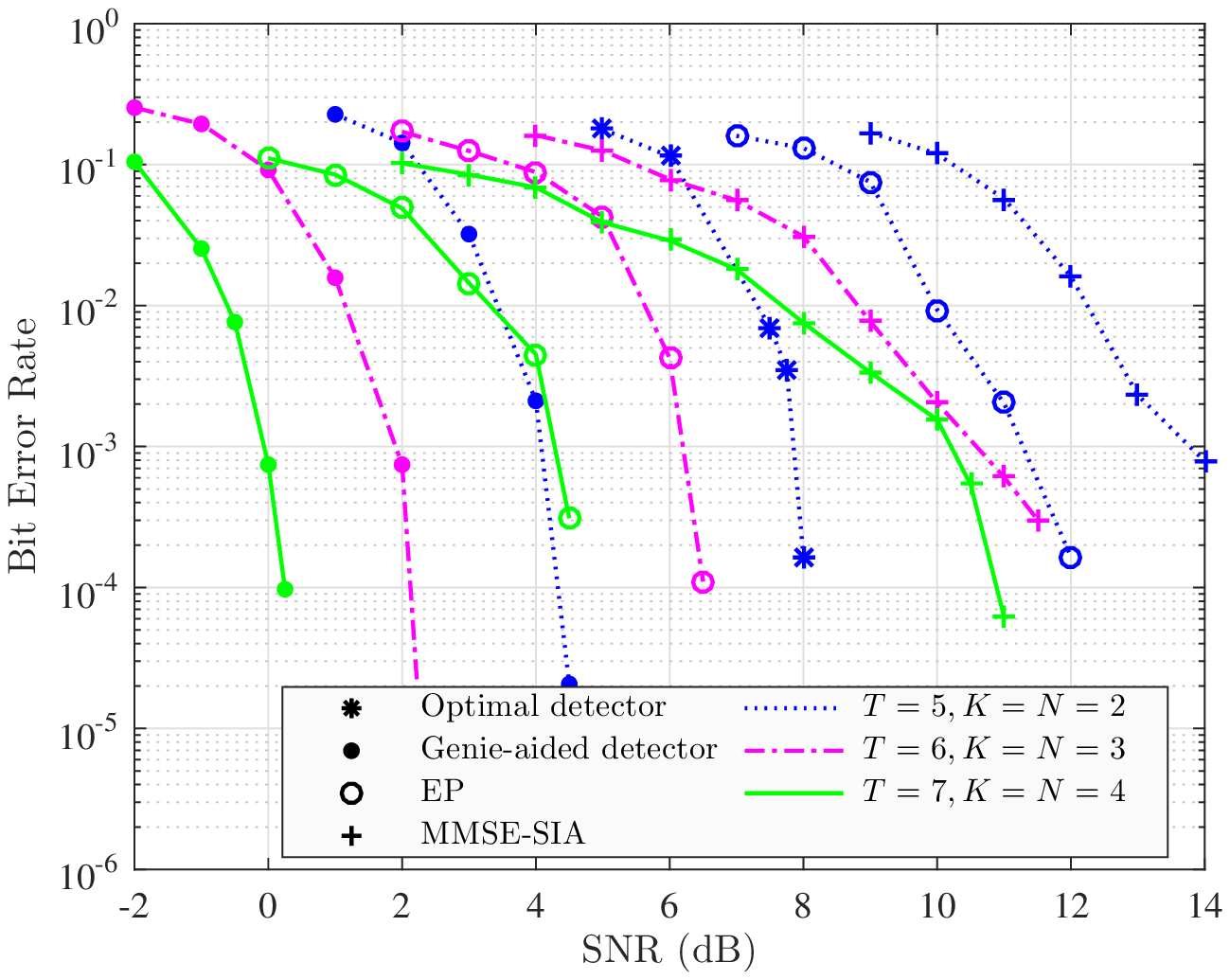}}
	\hspace{-.7cm}
	\caption{The bit error rate with turbo codes of EP, EPAK (with $t_0 = 2$), MMSE-\revise{SIA}, POCIS, and the optimal/genie-aided detector for $B = 8$ bits/symbol and $K = N$. 	}
	\label{fig:BER_turboK}
\end{figure}

Finally, in Fig.~\ref{fig:BER_turboB}, we consider $T=6$, $K=3$, $N =4$, and compare the BER with the same turbo code for different $B$. \revise{For $B = 5$, both EP and MMSE-SIA have performance close to the optimal detector. Under uncorrelated fading}, MMSE-\revise{SIA} can be slightly better than EP. This is due to the residual effect (after damping) of the phenomenon that all the mass of $\pi_{k1}\of{{i}_k}$ is concentrated on a possibly wrong symbol at early iterations, and EP may not be able to refine significantly the \revise{PMF} in the subsequent iterations if the constellation is sparse. This situation is not observed for $B = 8$, i.e., larger constellations. Also, as compared to the case $T = 6, K = 3, B = 8$ in Fig.~\ref{fig:BER_turboK}, the performance of MMSE-\revise{SIA} is significantly improved as the number of receive antennas increases from $N=3$ to $N=4$. As in the previous case, EPAK does not perform well.
\begin{figure}[!h] 
	\centering
	\subfigure[Uncorrelated fading]{\includegraphics[width=.51\textwidth]{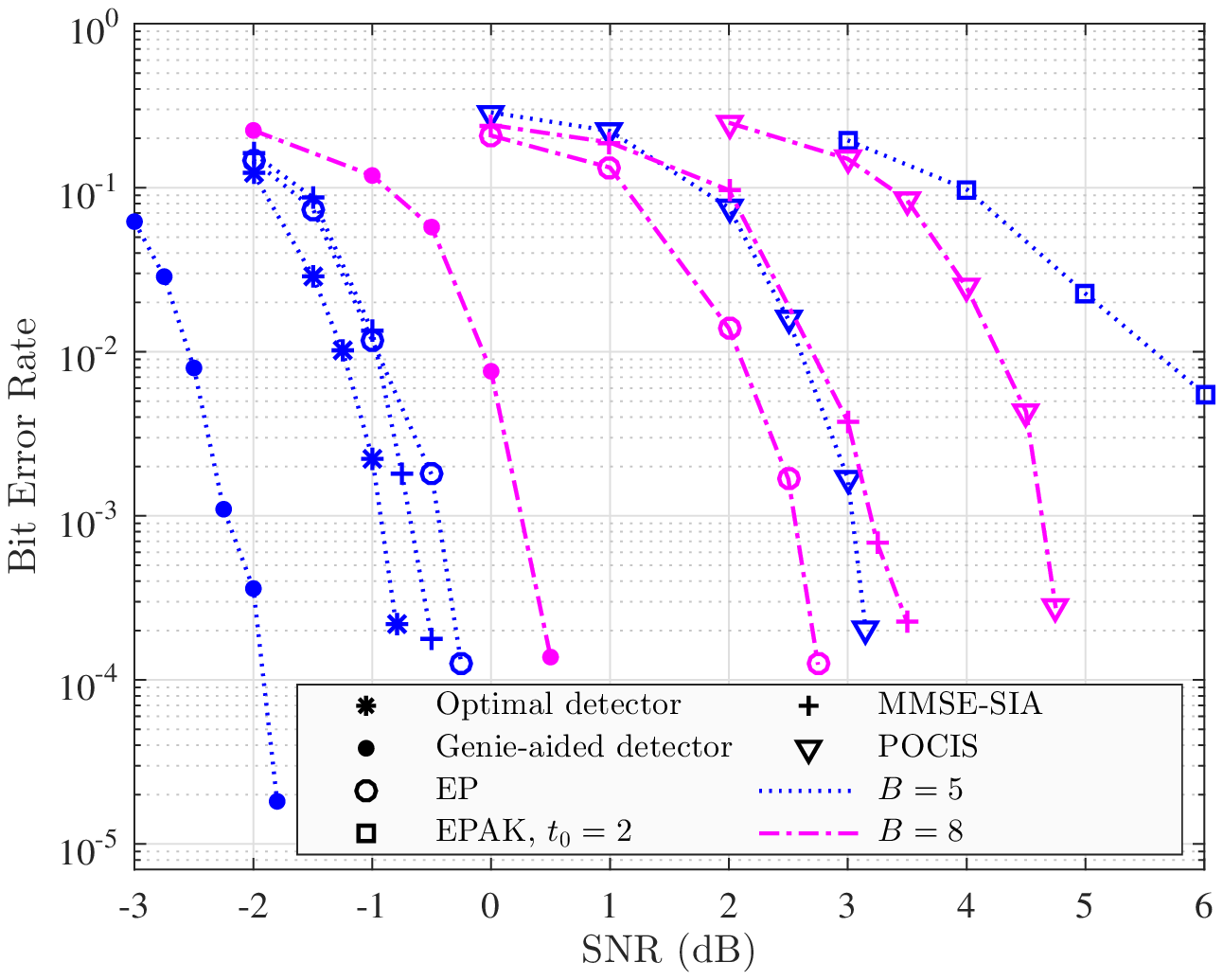}}
	\subfigure[Correlated fading]{\includegraphics[width=.51\textwidth]{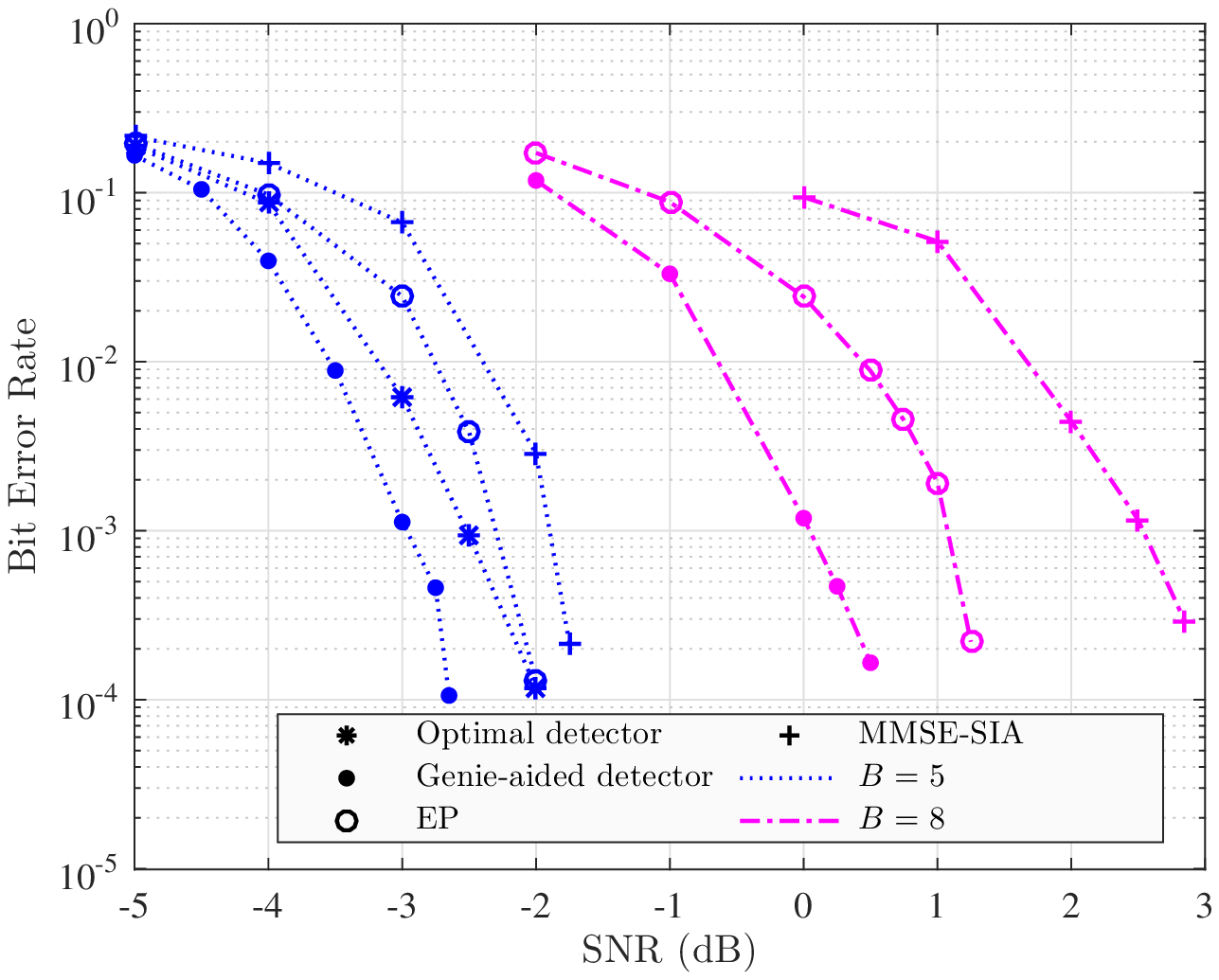}}
	\caption{The bit error rate with turbo codes of EP, EPAK (with $t_0 = 2$), MMSE-\revise{SIA}, POCIS, and the optimal/genie-aided detector for $T = 6$, $K =3$, and $N =4$.  	}
	\label{fig:BER_turboB} 
\end{figure}
\section{Conclusion} \label{sec:conclusion} 
We proposed an expectation propagation~(EP) based scheme and two simplifications (EPAK and MMSE-\revise{SIA}) of this scheme for multi-user detection in non-coherent SIMO multiple access channel \revise{with spatially correlated Rayleigh fading}. EP and MMSE-\revise{SIA} are shown to achieve good performance in terms of mismatched sum-rate, symbol error rate when they are used for hard detection, and bit error rate when they are used for soft-input soft-output channel decoding. EPAK \revise{has acceptable performance with uncorrelated fading}. It performs well for hard symbol detection but inadequately for soft-output detection. While MMSE-\revise{SIA} and EPAK have lower complexity than EP, the performance gain of EP with respect to MMSE-\revise{SIA} and EPAK is more significant when the number of users and/or the constellation size increase. \revise{Possible extensions of this work include considering more complicated fading models and analyzing theoretically the performance of EP for non-coherent reception.}

\appendix
\subsection{Properties of the Gaussian PDF} \label{app:gaussian_PDF}
\begin{lemma}\label{lemma:Gmult}
		Let $\rvVec{x}$ be an $n$-dimensional complex Gaussian vector. 
		It holds that
		\begin{enumerate}
			\item $\Nc(\xv;\muv,\Sigmam) = \Nc(\xv+\yv;\muv-\yv,\Sigmam)$ for $\yv \in \CC^n$;
			\item Gaussian PDF multiplication rule:
			$\Nc(\xv;\muv_1,\Sigmam_1)
			\Nc(\xv;\muv_2,\Sigmam_2)
			= \Nc(\xv;\muv\lnew,\Sigmam\lnew) 
			\Nc(\mathbf{0};\muv_1-\muv_2,\Sigmam_1+\Sigmam_2)$, 
			where $\Sigmam\lnew
			\defeq \big(\Sigmam_1^{-1}+\Sigmam_2^{-1}\big)^{-1}$ and $\muv\lnew
			\defeq \Sigmam\lnew
			\big(\Sigmam_1^{-1}\muv_1+\Sigmam_2^{-1}\muv_2\big)$.
	\end{enumerate}
\end{lemma}
\begin{proof}\revise{ 
		The first part follows readily from the definition of $\Nc(\xv; \muv,\Sigmam)$. The complex Gaussian PDF multiplication rule is a straightforward generalization of the real counterpart~\cite{Bromiley2003products}. 
	}
\end{proof}
\subsection{Proof of Proposition~\ref{prop:pnew}} \label{proof:pnew}
Using the natural logarithm for the KL divergence, we derive
\begin{align}
	&D\big( q_\alpha(\xv) \big\| \underline{p}(\xv)\big) 
= \int q_\alpha(\xv) \ln \frac{q_\alpha(\xv)}{ \prod_\beta \underline{p}_\beta(\xv_\beta) } \dif\xv \\
&= \sum_\beta \int q_\alpha(\xv) \ln \frac{1}{\underline{p}_\beta(\xv_\beta) } \dif\xv + \const \\
&= \sum_{\beta\in\Nalpha} \!\int\! q_\alpha(\xv) \ln \frac{1}{\underline{p}_\beta(\xv_\beta) } \dif\xv  
+ \sum_{\beta\notin\Nalpha} \!\int\! q_\alpha(\xv) \ln \frac{1}{\underline{p}_\beta(\xv_\beta) } \dif\xv 
\nonumber \\ &\quad  + \const 
\end{align}
\begin{align}
&= \sum_{\beta\in\Nalpha}  \int  q_\alpha(\xv) \ln \frac{1}{\underline{p}_\beta(\xv_\beta) } \dif\xv 
\nonumber \\ &\quad 
+  \sum_{\beta\notin\Nalpha}  \int  \hat{p}_\beta(\xv_\beta) \ln \frac{1}{\underline{p}_\beta(\xv_\beta) } \dif\xv_\beta + \const \label{eq:tmp381} \\
&= - \!\sum_{\beta\in\Nalpha}  \!\int\!  q_\alpha(\xv) \left[\underline{\gammav}_\beta^\T\phiv(\xv_\beta)  \!-\!  A_\beta(\underline{\gammav}_\beta)\right]\dif\xv 
+\!  \sum_{\beta\notin\Nalpha}\!  D\big(\hat{p}_\beta\big\|\underline{p}_\beta\big)
\nonumber \\ &\quad 	+ \const \label{eq:tmp384}\\
&= \sum_{\beta\in\Nalpha} \left[ A_\beta(\underline{\gammav}_\beta) - \underline{\gammav}_\beta^\T \E_{q_\alpha}\big[\phiv(\xv_\beta)\big] \right]
+  \sum_{\beta\notin\Nalpha} D\big(\hat{p}_\beta\big\|\underline{p}_\beta\big)
+ \const,
\quad \label{eq:KLalf1}
\end{align}
where~\eqref{eq:tmp381} follows from 
$
q_{\alpha}(\xv) = \frac{\psi_\alpha(\xv_\alpha)}{m_\alpha(\xv_\alpha)}
\Big[\prod_{\beta\in\Nalpha} \hat{p}_\beta(\xv_\beta) \Big] 
\Big[\prod_{\beta\notin\Nalpha} \hat{p}_\beta(\xv_\beta) \Big] ,
$
and \eqref{eq:tmp384} follows from \eqref{eq:expfam}.
From \eqref{eq:KLalf1}, we can see that the optimization~\eqref{eq:pnew} of $\underline{p}$ decouples over $\underline{p}_\beta$, and the optimal distribution can be expressed as $\hat{p}\new_\alpha(\xv) = \prod_\beta \hat{p}\new_{\alpha,\beta}(\xv_\beta)$. 
For $ \beta\notin\Nalpha$, the minimum of $D\big(\hat{p}_\beta\big\|\underline{p}_\beta\big)$ is simply $0$ and achieved with $\hat{p}\new_{\alpha,\beta}(\xv_\beta)={\hat{p}}_\beta(\xv_\beta)$. 
For $\beta\in\Nalpha$, since the log-partition function $A_\beta(\underline{\gammav}_\beta)$ is convex in $\underline{\gammav}_\beta$ (see, e.g.,~\cite[Lemma~1]{Wainwright2005logpartition}), the minimum of $ A_\beta(\underline{\gammav}_\beta) - \underline{\gammav}_\beta^\T \E_{q_\alpha}\big[\phiv(\xv_\beta)\big]$ is achieved at the value of $\uvec{\gammav}_\beta$ where its gradient is zero. Using the well-known property of the log-partition function, $\nabla_{\gammav} A_\beta(\gammav)=\E_{\hat{p}_\beta}[\phiv_\beta(\gammav)]$, we get that the zero-gradient equation is equivalent to the moment matching criterion $\E_{\hat{p}\new_{\alpha,\beta}}[\phiv_\beta(\xv_\beta) = \E_{q_\alpha}[\phiv_\beta(\xv_\beta)]$. 

\bibliographystyle{IEEEtran}
\bibliography{IEEEabrv,./biblio}

\end{document}